%% file: appendix-arxiv.tex
\documentclass [11pt] {article}
\usepackage{fullpage}

\usepackage{graphics}
\usepackage[dvips]{epsfig}
\usepackage{amsmath,amsthm}
\usepackage{amssymb}
\usepackage{amsfonts}
\usepackage{graphicx}
\usepackage{xspace}
\usepackage{algorithm}
\usepackage{comment}
\usepackage{algpseudocode}
\usepackage{xcolor}
\usepackage{cite}
\usepackage{array}
\usepackage{color}

\usepackage{enumitem}

\newtheorem{theorem}{Theorem}[section]
\newtheorem{lemma}[theorem]{Lemma}
\newtheorem{corollary}[theorem]{Corollary}
\newtheorem{claim}{Claim}
\newtheorem{proposition}[theorem]{Proposition}
\newtheorem{definition}[theorem]{Definition}

\newtheorem{observation}[theorem]{Observation}

\newcommand{\canonDRIP}{\ensuremath{\mathcal{D}_G}\xspace}
\newcommand{\arbDRIP}{D}

\newcommand{\canonHist}[2]{\mathcal{H}_{#1,\canonDRIP}[0 \ldots #2]}
\newcommand{\DHist}[2]{\mathcal{H}_{#1,\arbDRIP}[0 \ldots #2]}

\newcommand{\N}{\mathbb{N}}

\newcommand{\cH}{{\cal H}}

\newcommand{\PARTITION}{\texttt{Refine}\xspace}
\newcommand{\classifier}{\texttt{Classifier}\xspace}
\newcommand{\initgaug}{\texttt{Init-Aug}\xspace}
\newcommand{\partitioner}{\texttt{Partitioner}\xspace}
\newcommand{\CLASS}{\textrm{CLASS}\xspace}
\newcommand{\TAG}{\textrm{LBL}\xspace}

\begin{document}
	\def\thefootnote{\fnsymbol{footnote}}
	
	\title{{\bf Deterministic Leader Election\\ in Anonymous Radio Networks}}

	\author{
		Avery Miller\footnotemark[1]
		\and Andrzej Pelc\footnotemark[2]
		\and Ram Narayan Yadav\footnotemark[3]
	}
	
	\footnotetext[1]{Department of Computer Science, University of Manitoba, Winnipeg, Manitoba, R3T 2N2, Canada. {\tt avery.miller@umanitoba.ca}. Supported by NSERC Discovery Grant RGPIN--2017--05936.}
	\footnotetext[2]{D\'epartement d'informatique, Universit\'e du Qu\'ebec en Outaouais, Gatineau,
		Qu\'ebec J8X 3X7, Canada. {\tt pelc@uqo.ca}. Partially supported by NSERC Discovery Grant RGPIN--2018--03899
		and by the Research Chair in Distributed Computing at the
		Universit\'e du Qu\'ebec en Outaouais.}
	\footnotetext[3]{Department of Computer Science and Engineering, Institute of Infrastructure Technology Research and Management (IITRAM), Gujarat, India. {\tt narayanram.1988@gmail.com}}
	
	\maketitle
	
	\thispagestyle{empty}

\begin{abstract}
Leader election is a fundamental task in distributed computing. It is a symmetry breaking problem, calling for one node of the network to become the {\em leader}, and for all other nodes to become {\em non-leaders}. We consider leader election in anonymous radio networks modeled as simple undirected connected graphs. Nodes communicate in synchronous rounds. In each round, a node can either transmit a message to all its neighbours, or stay silent and listen. A node $v$ hears a message from a neighbour $w$ in a given round if $v$ listens in this round and if $w$ is its only neighbour transmitting in this round. If $v$ listens in a round in which more than one neighbour transmits then $v$ hears noise that is different from any message and different from silence.

We assume that nodes are identical (anonymous) and execute the same deterministic algorithm. Under this scenario, symmetry can be broken only in one way: by different wake-up times of the nodes. In which situations is it possible to break symmetry and elect a leader using time as symmetry breaker? In order to answer this question, we consider {\em configurations}. A configuration is the underlying graph with nodes tagged by non-negative integers with the following meaning. A node can either wake up spontaneously in the round shown on its tag, according to some global clock, or can be woken up hearing a message sent by one of its already awoken neighbours. The local clock of a node starts at its wakeup and nodes do not have access to the global clock determining their tags. A configuration is {\em feasible} if there exists a distributed algorithm that elects a leader for this configuration.

Our main result is a complete algorithmic characterization of feasible configurations. More precisely, we design a centralized decision algorithm, working in polynomial time, whose input is a configuration and which decides if the configuration is feasible. Using this algorithm we also provide a dedicated deterministic distributed leader election algorithm for each feasible configuration that elects a leader for this configuration in time $O(n^2\sigma)$, where $n$ is the number of nodes and $\sigma$ is the difference between the largest and smallest tag of the configuration. We then ask the question if there exists a universal  deterministic distributed algorithm electing a leader for all feasible configurations. The answer turns out to be no, and we show that such a universal algorithm cannot exist even for the class of 4-node feasible configurations. We also prove that a distributed version of our decision algorithm
%(i.e., a distributed algorithm which when run on a configuration would make some or all nodes correctly decide if the configuration is feasible) 
cannot exist.

\vspace*{0.5cm}

\noindent
{\bf keywords:}  leader election, anonymous radio network, graph, algorithm

\end{abstract}

\pagebreak

%%%%%%%%%%%%%%%%%%%%%%%%%%%%%%%%%%%%
\section{Introduction}
%%%%%%%%%%%%%%%%%%%%%%%%%%%%%%%%%%%%

\subsection{The model and the problem}

Leader election is a fundamental distributed task involving symmetry breaking: initially all nodes of a network have the same
status {\em non-leader} and the goal is for all nodes but one to keep this status 
and for the remaining single node to get
the status {\em leader}.
The problem of leader election was first formulated in \cite{LL} in the study of local area token ring networks, where, at all times, exactly one node (the owner of a circulating token) is allowed to initiate
communication. When the token is accidentally lost, a leader must be elected as the initial owner of the token. 

 We consider the task of leader election in radio networks, modeled as simple undirected connected graphs. 
 %In the sequel, we use the word {\em graph} in this sense, and we consider the notions of {\em network} and {\em graph} as synonyms.
 A node can either wake up spontaneously, or can be woken up hearing a message sent by one of its already awaken neighbours.
Nodes communicate in synchronous rounds. In each round, a node can either transmit a message to all its neighbours, or stay silent and listen. At the receiving end, a node $v$ hears a message from a neighbour $w$ in a given round if $v$ listens in this round and if $w$ is its only neighbour that transmits in this round.
If more than one neighbour of a node $v$ transmits in a given round, we say that a {\em collision} occurs at $v$. We make the well-established and practically motivated assumption of the capability of {\em collision detection} (cf., e.g., \cite{KP,TM,Wil}): if a node $v$ listens and a collision occurs at $v$, then $v$ hears noise that is different from any message and also different from silence. Finally, a node that transmits in a given round does not hear anything.

We assume that nodes are anonymous (identical) and execute the same deterministic algorithm. 
%Nodes do not have any {\em a priori} knowledge.  
Under this very weak scenario, symmetry can be broken only in one way: by different wake-up times of the nodes. Indeed, if all nodes wake up in the same round, in each subsequent round they will either all transmit or all listen, and no message will be ever heard.
In which situations is it possible to break symmetry and elect a leader using wake-up time as symmetry breaker? In order to answer this question, we consider {\em configurations}. A configuration is defined as the graph underlying the radio network with nodes tagged by non-negative integers indicating the round of spontaneous wakeup of the node, according to some global clock. Hence a node either wakes up spontaneously in the round indicated by its wakeup tag, or wakes up in an earlier round, if it hears a message in this round. 
The local clock of a node has value 0 in its wakeup round, and a node starts executing its algorithm in local round 1. Nodes do not have access to the global clock determining their tags. A configuration is {\em feasible} if there exists a deterministic distributed algorithm that elects a leader for this configuration. 

Our research is motivated by the following question: Which are the feasible configurations and does there exist a universal deterministic algorithm electing a leader in all of them?

Deterministic leader election in anonymous networks is a difficult task, even in the model of wired message passing networks because nodes do not have distinct labels permitting us to immediately break symmetry between them.
However, in wired networks, where distinct port numbers are available at each node, enabling it, for example, to learn its degree, and, on the other hand, messages are guaranteed to arrive at neighbours, regardless of time rounds in which nodes transmit,  leader election can be based on the topological structure of the network:  nodes  can relay their neighbourhoods of increasing radii, learning in this way asymmetries of the network topology, which can eventually serve leader election. In this case, differences of wake-up times do not have to be exploited, and symmetry breaking can be done exclusively on the basis of graph structure considerations. In contrast, in anonymous radio networks, differences of wake up times must be involved because otherwise, as mentioned above, no communication between nodes can be achieved. Hence, in a sense, the scenario of anonymous radio networks is the most adverse scenario for symmetry breaking, and our research can be seen as investigating if and when symmetry breaking is at all possible in this extreme case.

\subsection{Our results}

Our main result is a complete algorithmic characterization of feasible configurations. More precisely, we design a centralized decision algorithm whose input is a configuration and which decides if the configuration is feasible. Our algorithm works in time $O(n^3 \Delta)$, where $n$ is the number of nodes and $\Delta$ is the maximum degree. Using this algorithm we also provide a dedicated distributed leader election algorithm for each feasible configuration that elects a leader for this configuration in time $O(n^2\sigma)$, where $\sigma$ is the difference between the largest and smallest tag of the configuration. On the negative side we prove that time complexity $o(\sigma +n)$ cannot be achieved for  some configurations. We then ask the question if there exists a universal  deterministic distributed algorithm electing a leader for all feasible configurations. The answer turns out to be no, and we show that such a universal algorithm cannot exist even for the class of 4-node feasible configurations. We also prove that a distributed version of our decision algorithm
(i.e., a deterministic distributed algorithm which, when run on any feasible configuration, would make all nodes say ``yes'', and when run on any unfeasible configuration would make some node say ``no'') 
cannot exist. 

%As an illustration of our algorithmic characterization of feasible configurations, we consider the case when the underlying graph is a line and we give a concise combinatorial characterization of configuration feasibility in this case. It turns out that a linear configuration is feasible if and only if it cannot be represented as a concatenation of blocks $(A A^R)^n$ or $(A A^R)^nA$, for any positive integer $n$, where $A$ is any block
%(i.e., a sequence of non-negative integers) of positive length, and $A^R$ denotes the reverse of block $A$.

%%%%%%%%%%%%%%%%%%%%%%%%%%%%%%%%%%%%
\subsection{Related work}
%%%%%%%%%%%%%%%%%%%%%%%%%%%%%%%%%%%%

{\bf Leader election in labeled networks.}
Leader election is a classic topic in distributed computing, and has been
widely studied in the early history of this domain (cf. \cite{Ly}). 
The problem of leader election was first mentioned in \cite{LL}. 
Early papers on leader election focused on the scenario 
where all nodes have distinct labels. Initially, it was investigated for rings in the message passing model.
A synchronous algorithm based on label comparisons was given in \cite{HS}. It used 
$O(n \log n)$ messages.  In \cite{FL} it was proved that
this complexity cannot be improved for comparison-based algorithms. On the other hand, the authors showed
a leader election algorithm using only a linear number of messages but requiring very large running time.
An asynchronous algorithm using $O(n \log n)$ messages was given, e.g., in \cite{P}, and
the optimality of this message complexity was shown in \cite{B}. 
In \cite{CMRZ}, the authors investigated the time of leader election in point-to-point networks whose nodes have logarithmic labels, establishing optimal election time
under the assumption that messages are of constant size. 
Leader election was also investigated in the radio communication model,
both in the deterministic \cite{JKZ,KP} and in the randomized \cite{Wil} scenarios.
 In \cite{DB,IRSVWW}, the task of leader election was studied in the context of dynamic networks.

\noindent
{\bf Leader election in anonymous networks.}
Many authors \cite{An,AtSn,ASW,BSVCGS,BV,YK2,YK3} studied leader election
in anonymous networks. In particular, \cite{BSVCGS,YK3} characterized message-passing networks in which
leader election is feasible. In \cite{YK2}, the authors studied
the problem of leader election in general networks, under the assumption that node labels exist but are
not unique. They characterized networks in which leader election can be performed and gave an algorithm
which achieves election when it is feasible. 
%They assume that the number of nodes of the network is known to all nodes. 
In  \cite{DoPe,FKKLS},  the authors
studied message complexity of leader election in rings with possibly
nonunique labels. 
%Characterizations of feasible instances for leader election were provided in~\cite{C,CM}.
Memory needed for leader election in unlabeled networks was studied in \cite{FP}. 
In \cite{DP1}, the authors investigated the feasibility of leader election among anonymous agents that
navigate in a network in an asynchronous way. In \cite{GMP} leader election was studied in the context of the size of advice needed to accomplish it in a given time. Other computing tasks in anonymous networks were considered, e.g., in \cite{BSVCGS,BV,DP,YK3}.

\noindent
{\bf Leader election in radio networks.}
Algorithmic problems in radio networks modeled as graphs were studied for such tasks as broadcasting \cite{CGR,GPX}, gossiping \cite{CGR} and leader election
\cite{KP}. In some cases \cite{CGR,CzD}, the topology of the network was unknown, in others \cite{CGOR,EK3,GM,GPX}, nodes were assumed to have a labeled map of the network and could situate themselves in it.

Most of the results on leader election in the radio model concern {\em single-hop} 
networks of known size $n$.
Some of these results were originally obtained for other distributed problems but have corollaries for leader election.
For the time of deterministic leader election without collision detection, 
the complexity
$O(n\log n)$ follows from \cite{CMS}. 
A constructive upper bound $O(n\cdot polylog(n))$ follows from \cite{Ind}.
For the time of deterministic algorithms with collision detection, 
matching bounds are also known:
$\Omega(\log n)$ follows from~\cite{GW}, and $O(\log n)$ follows from \cite{Cap,Hay,TM}.
For the expected time of randomized algorithms without collision detection, 
the same matching bounds are known:
$\Omega(\log n)$ follows from \cite{KM} and $O(\log n)$ from~\cite{BGI}.
Finally, randomized leader election with collision detection can be done faster:
matching bounds $\Omega(\log\log n)$ (for fair protocols) and $O(\log\log n)$
on the expected time
were proved in~\cite{Wil}.

For leader election in arbitrary radio networks results are less complete. 
Deterministic algorithms without collision detection were proposed in \cite{CKP,CzD}: the algorithm from \cite{CKP} works in time $O(n\log ^{3/2}n \sqrt{\log\log n})$
and the algorithm from \cite{CzD} works in time $O(n \log n \log D \log \log D)$.
In \cite{KP} it was shown how to elect a leader in arbitrary radio networks in time $O(n)$, if collision detection is assumed.
In \cite {CzD2} the authors gave a randomized leader election algorithm, without collision detection, working in time $O(D \log n / \log D + polylog (n))$ with high probability.

In all the above papers, results concerning deterministic leader election in radio networks assumed that nodes have distinct labels.
To the best of our knowledge, 
no results are published for deterministic
leader election in anonymous radio networks.

\section{Terminology and Notation}
\input{terminology.tex}

\section{Efficient Classification of Feasible Configurations}
\input{classifier.tex}

\subsection{Correctness of \classifier}
\input{characterization.tex}

\section{Negative results}

In this section, we prove lower bounds on the complexity of dedicated leader election algorithms for feasible configurations, and prove impossibility results concerning universal leader election and distributed decision algorithms for anonymous radio networks.  Our first negative result is a $\Omega(n)$ lower bound on the complexity of leader election, even for some configurations with bounded span.

\begin{proposition}
There exists an infinite class of feasible configurations with span $\sigma=1$, such that, for each configuration $G$ of this class, every dedicated leader election algorithm for $G$ takes time $\Omega(n)$,
where $n$ is the size of the configuration.
\end{proposition}

\begin{proof}
Consider  the class of linear configurations $G_m$ with nodes $a_1,\dots, a_m,b_1,\dots ,b_{2m+1},c_m,\dots, c_1$, listed from left to right, for $m \geq 2$. For all $i \in \{1,\ldots,m\}$, the wakeup tags of nodes $a_i$ and $c_i$ are 0. For all $i \in \{1,\ldots,2m+1\}$, the wakeup tags of nodes $b_i$ are 1. By Lemma \ref{YesCorrectness}, all configurations $G_m$ are feasible: indeed, when \classifier is executed with input $G_m$, the central node $b_{m+1}$ will be in a one-element equivalence class after $m$ iterations.  Consider any leader election algorithm for configuration $G_m$. For any local round and any $i \in \{1,\ldots,m\}$, the history of nodes $a_i$ and $c_i$ is the same, and the history of nodes $b_i$ and $b_{2m+2-i}$ is the same, due to the symmetry of the configuration. Moreover, for any local round $t<m-1$, the history of nodes $b_m, b_{m+1}, b_{m+2}$ is the same: either all of them transmit or all of them listen and hear silence in each of these rounds. Hence, in all local rounds $t<m-1$ leader election is impossible. Since $m \in \Theta(n)$, this concludes the proof.
\end{proof}

In our remaining negative results, we will make use of the following class of configurations. For each $m \geq 1$, denote by $H_m$ the linear configuration of size 4 consisting of nodes $a,b,c,d$, listed from left to right, with the following wakeup tags: nodes $b$ and $c$ have tags 0, node $a$ has tag $m$ and node $d$ has tag $m+1$. The following lemma gives a lower bound on the number of rounds needed to solve leader election in such configurations.

\begin{lemma}\label{lemma neg}
Each configuration $H_m$ is feasible, and every leader election algorithm for $H_m$ takes time at least $m$.
\end{lemma}

\begin{proof}
By Lemma \ref{YesCorrectness}, all configurations $H_m$ are feasible: indeed, when \classifier is executed with input $H_m$, each of the four nodes will be in a one-element class after iteration 1. Suppose that there exists a leader election algorithm $\cal  A$ for configuration $H_m$ working in time less than $m$. In the execution of $\cal  A$, nodes $b$ and $c$ with tag 0 must send their first message before round $m$, otherwise all nodes would have the same history before round $m$ (every entry equal to $(\emptyset)$) and leader election could not be correctly achieved. Further, nodes $b$ and $c$ send their first message in the same round, as they wake up in the same global round and have the same history up to that round (every history entry equal to $(\emptyset)$). Suppose that nodes $b$ and $c$ send their first message in round $t<m$. Nodes $a$ and $d$ are woken up by these messages, and, from round $t$ onward, the histories of nodes $a$ and $d$ are the same. The histories of nodes $b$ and $c$ were the same up to round $t-1$ and will also be the same from round $t$ onward. This follows by induction on the round number. Hence algorithm  $\cal  A$ cannot correctly elect a leader.
\end{proof}

Lemma \ref{lemma neg} implies our second negative result, which is a $\Omega(\sigma)$-round lower bound on the complexity of leader election, even for some configurations of bounded size.

\begin{proposition}
There exists an infinite class of feasible configurations of size $n=4$, such that, for each configuration $G$ of this class, every dedicated leader election algorithm for $G$ takes time $\Omega(\sigma)$,
where $\sigma$ is the span of the configuration.
\end{proposition}

We now consider the question whether there exists a universal distributed algorithm that elects a leader for all feasible configurations. Our next result shows that the answer is no. In fact, even knowing the size of the configuration cannot help.

\begin{proposition}
There is no universal distributed algorithm that elects a leader for all feasible configurations of size 4.
\end{proposition}

\begin{proof}
Suppose that such a universal algorithm $\cal U$ exists. If no node ever sends a message then leader election is impossible.
Consider the configurations $H_m$ with $m \geq 1$. By Lemma \ref{lemma neg}, they are all feasible.  Suppose that $t$ is the first global round when nodes with tag 0 send a message. Both nodes with tag 0 will send the same first message, as they both have the same history up to round $t$ (every entry equal to $(\emptyset)$).
Consider configuration $H_{t+1}$. Nodes $a$ and $d$ are woken up  by the first message of nodes $b$ and $c$ respectively, and, for all rounds after wakeup, the histories of nodes $a$ and $d$ are the same. The histories of nodes $b$ and $c$ were the same up to round $t-1$ and will also be the same from round $t$ onward, due to symmetry. Hence $\cal U$ does not correctly elect a leader on configuration $H_{t+1}$, which is a contradiction.
\end{proof}

Finally, we consider the question whether feasibility of a configuration can be decided in a distributed way. Algorithm {\tt Classifier} is a decision algorithm for the property of feasibility, but it is centralized: the configuration is given to it as input and the algorithm correctly outputs the decision. (Of course, such a centralized algorithm can be simulated in a distributed way if nodes get the configuration as input). A hypothetical distributed decision algorithm would work as follows, for all configurations: all nodes of a configuration output ``yes'' if the configuration is feasible, and at least one node outputs ``no'' if the configuration is not feasible. Our next result shows that such a distributed counterpart of Algorithm {\tt Classifier} cannot exist.

\begin{proposition}
If nodes have no a priori knowledge, there is no distributed algorithm that decides if a configuration is feasible.
\end{proposition}

\begin{proof}
Suppose that such a distributed decision algorithm $\cal D$ exists. We define a sequence of linear configurations $S_m$, for $m \geq 1$, as follows.
The nodes of $S_m$ are $a,b,c,d$, listed from left to right, with the following wakeup tags:
nodes $b$ and $c$ have tag 0, and nodes $a$ and $d$ have tag $m$. By Lemma \ref{NoCorrectness}, the configurations $S_m$ are not feasible: indeed, when \classifier is executed on $S_m$ for any $m \geq 1$, the partition of nodes into equivalence classes after iteration 2 will be the same as after iteration 1, and will consists of two classes with two elements each, so \classifier will output ``No". However, recall from Lemma  \ref{lemma neg} that the configurations $H_m$ are feasible for all $m \geq 1$.

Algorithm $\cal D$ must instruct the nodes to send some message, otherwise no correct decision can be made. Suppose that $t$ is the first round when nodes with tag 0 send a message.
Consider configurations $H_{t+1}$ and $S_{t+1}$. The history of each of the nodes  $a,b,c,d$ is the same in both these configurations, for all rounds. Hence, each of the nodes executing algorithm $\cal D$ must make the same decision when $\cal D$ terminates in configurations $H_{t+1}$ and $S_{t+1}$. This is a contradiction, as one of these configurations is feasible and the other one is not.
\end{proof}

\section{Conclusion}

We characterized the configurations for which leader election is possible for anonymous radio networks, which is a particularly difficult scenario for this task.
The characterization is done by a centralized decision algorithm accompanied by a dedicated distributed leader election algorithm for each feasible configuration.
We proved the nonexistence of a distributed algorithm deciding whether a configuration is feasible, and the nonexistence of a universal distributed leader election algorithm working for all feasible
configurations. Thus, in terms of feasibility, the problem of leader election in anonymous radio networks is completely solved.

As far as time complexity is concerned, two problems remain open in the context of this work. The first is the complexity of the centralized decision algorithm. 
Can the complexity $O(n^3 \Delta)$ of Algorithm {\tt Classifier} be improved?
What is the optimal time complexity of a centralized decision algorithm for this task? As for distributed dedicated leader election algorithms, our algorithm using the canonical DRIP for feasible configurations  works in time $O(n^2\sigma)$ and we proved the lower bound 
$\Omega(n+\sigma)$ on the complexity of dedicated leader election for some classes of feasible configurations. Hence a natural open problem is whether there exists a  $O(n+\sigma)$  dedicated leader election algorithm for each feasible configuration.

%\newpage

\end{document}

%% file: terminology.tex
\subsection{Configuration}
	
	A \emph{configuration} is an undirected graph where each node $v$ is tagged with a non-negative integer $t_v$.  A configuration represents a radio network in which each node $v$ wakes up in a global round $r \leq t_v$ if $v$ receives a message in global round $r$ (called a forced wakeup), or in global round $t_v$ otherwise (called a spontaneous wakeup). For a configuration $G$, the number of nodes of $G$ is called the {\em size} of $G$ and is denoted by $n$, and the difference between the largest and smallest wakeup tag is called the {\em span} of $G$ and is denoted by $\sigma$. Since nodes do not have access to the global clock determining the wakeup tags, we can assume without loss of generality that the smallest wakeup tag is 0, and hence the span is equal to the largest wakeup tag.
	
	\subsection{Distributed Radio Interaction Protocol (DRIP)}
	We now formally define the notion of a distributed communication protocol being executed by each node in a configuration $G$.  In each round $i \geq 1$ on its local clock, each node $v$ decides whether it will listen, transmit a message, or terminate its execution. This decision in each round depends on all of the knowledge the node knows so far, which we now define formally.
	
%	We will use the notion of a history vector $\cH_v$ of a node $v$, where $\cH_v[i]$ records the information gained by $v$  in local round $i\geq 0$. The value of $\cH_v[0]$ is defined to be $(\emptyset)$ if $v$ wakes up spontaneously, and otherwise is defined to be $(M)$ where $M$ is the message received by $v$ during the forced wakeup. 

For each $i \geq 0$, the \emph{history of node $v$ in round $i$}, denoted by $\cH_v[i]$, is defined to be: 
	\begin{itemize}
		\item $(\emptyset)$ if $v$ transmits in local round $i$, or, listens and receives no message in local round $i$, 
		\item $(M)$ if $v$ listens in local round $i$ and receives message $M$,  or, if $i=0$ and $v$ was woken up by message $M$,
		\item $(*)$ if $v$ listens in local round $i$ and a collision occurs at $v$.
	\end{itemize}
	This definition indicates that a node can distinguish if it wakes up spontaneously or by a message of a neighbour, and in the latter case it records the wakeup message in its history. 
	Consider an arbitrary function $D$ that takes as input a node's history vector up to some round, and outputs one of the following strings: $listen$, $transmit(M)$ for some string $M$, or $terminate$. A \emph{distributed radio interaction protocol} (DRIP) is defined as such a function $D$ in the following way: each node $v$, in each round $i \geq 1$ on its local clock, computes $D(\cH_v[0\ldots i-1])$ and performs the action described by the value of $D$. We require that each node eventually terminates permanently, i.e., for each node $v$, there exists an $i \geq 1$ such that $D(\cH_v[0\ldots i-1]) = terminate$, and $D(\cH_v[0\ldots i'-1]) = terminate$ for all $i' \geq i$. 	In the execution of any DRIP $D$, for any given $i \geq 0$, we denote the history of a node $v$ up to local round $i$ by $\DHist{v}{i}$.

	A \emph{patient DRIP} is a DRIP such that no node transmits in global rounds $0,\ldots,\sigma$. Since all wakeup tags are in this range, it follows that, when executing a patient DRIP, all nodes wake up spontaneously in the global round equal to their wakeup tag. Therefore, we have a reliable way of converting between local clock values and the global round number, as provided in the following result.
	
	\begin{proposition}\label{prop:localconversion}
		For any patient DRIP $D$, any two nodes $v,w$ executing $D$, and any $i \geq 0$, local round $i$ at node $v$ occurs in the same global round as local round $i - (t_w - t_v)$ at node $w$.
	\end{proposition}
	\begin{proof}
		Let $r$ denote the global round number corresponding to local round $i$ at node $v$. Let $r^{(w)}$ denote the local round number at $w$ that corresponds to global round $r$. Since $D$ is a patient DRIP, both $v$ and $w$ wake up spontaneously in global rounds $t_v$ and $t_w$, respectively. It follows that $r = t_v+i$ and $r=t_w+r^{(w)}$. Setting $t_v+i = t_w+r^{(w)}$, and solving for $r^{(w)}$ gives the desired result.
	\end{proof}

		\subsection{Leader Election Algorithm}
	For any DRIP $D$ and any node $v$, let $done_{v,D}$ denote the first round $i$ for which $D(\cH_{v,D}[0\ldots i-1]) = terminate$. When it is clear from the context which DRIP is being executed, we will just write $done_v$. A \emph{decision function $f$ for a DRIP $D$} takes as input a node's history vector induced by the execution of $D$, i.e., $\cH_{v,D}[0\ldots done_{v,D}]$, and outputs a 0 or 1. A \emph{dedicated leader election algorithm for configuration $G$} is a DRIP $D$ along with a decision function $f$ for $D$ such that $f(\cH_{v,D}[0\ldots done_{v,D}]) = 1$ for exactly one node $v \in G$. A configuration $G$ is {\em feasible} if there exists a dedicated leader election algorithm for $G$.

%% file: classifier.tex
%Denote by $N(v)$ the set of nodes at distance at most 1 from $v$. 

In this section, we set out to define a procedure that determines whether or not a given configuration $G$ is feasible. The challenge is to make such a procedure efficient, since trying every potential leader election algorithm in a brute-force manner is prohibitively expensive. We describe a centralized algorithm called \classifier that, when given as input a configuration $G$ with $n$ nodes, decides whether or not $G$ is feasible in time polynomial in $n$. Further, if $G$ is feasible, we can explicitly produce a distributed dedicated leader election algorithm for $G$ without any additional computation.
%In particular, given a configuration $G$ as input, \classifier determines whether or not the canonical DRIP \canonDRIP can be used to solve leader election when executed by the nodes of configuration $G$. The importance of characterizing feasible configurations in terms of the canonical DRIP becomes apparent: it would not be possible to simulate every possible leader election algorithm on configuration $G$ in an efficient way, but we can efficiently simulate the outcome of the execution of \canonDRIP.
\subsection{Definition of \classifier}
The high-level idea behind \classifier is to carry out the following phase-based algorithm. At the start of each phase, the nodes of $G$ are partitioned into equivalence classes, where nodes in the same class have the same history so far. Phase $P_0$ at each node consists of one round: its spontaneous wakeup round. All nodes are placed in the same class at the end of phase $P_0$. For each phase $i \geq 1$, we denote by $numClasses_i$ the number of equivalence classes at the start of phase $i$. Phase $i$ consists of $numClasses_i$ ``transmission blocks", where each transmission block consists of $2\sigma+1$ rounds. The idea is that each class is assigned its own transmission block, and we assume that each node in each class $k$ transmits in its local round $\sigma+1$ of transmission block $k$ (recalling that, due to different wake-up times, local round $\sigma+1$ at different nodes can correspond to different global rounds). In particular, we determine the history of each node $v$ during phase $i$ by considering in which equivalence class each of its neighbours $w$ started the phase and the relative wake-up times of $v$ and $w$. Finally, once the history for phase $i$ has been computed for each node $v$ in $G$, the equivalence classes are refined by comparing these histories, and the algorithm proceeds to the next phase. This is repeated until either: there exists a class consisting of exactly one node $v$ (in which case, \classifier outputs ``Yes" since $v$ can be chosen as the leader), or, no such class exists, and there are two consecutive phases with no changes in the partition (in which case, \classifier outputs ``No" since no further changes will ever occur and there is no possible leader). While it seems that \classifier will only determine whether or not the specific algorithm described above can solve leader election in configuration $G$, we will later show that this is sufficient in order to determine the feasibility of $G$.

We now give a detailed description of \classifier. Each node $v$ in $G$ is augmented with the following variables:
\begin{itemize}
	\item A $v_{\CLASS}$ variable that keeps track of which equivalence class the node belongs to.
	\item A label $v_{\TAG}$ that will be used to determine if two nodes in the same equivalence class should stay in the same equivalence class after the current phase. Essentially, at the start of each phase, this label represents what node $v$ heard during the previous phase.
\end{itemize}
The configuration is also augmented with a variable $numClasses_G$ that keeps track of the number of different classes in the partition, as well as a list $reps$ of representative nodes, one for each equivalence class in the partition.
We denote by $G_{aug}$ the augmented version of configuration $G$. Recall that each node $v$ of a configuration $G$ is labeled with its wakeup tag $t_v$, so the wakeup tag values are also available in the augmented configuration $G_{aug}$. Finally, we fix an arbitrary ordering of the vertices of $G_{aug}$ so that loops of the form ``{\bf for all} $v \in G_{aug}$" always iterate through the nodes in the same order. Algorithm \ref{initgaug} gives the pseudocode describing how the augmented configuration $G_{aug}$ is initialized.

\begin{algorithm}[H]
	\small
	\caption{\initgaug, input is configuration $G$, where each node $v$ in $G$ is labeled with its wake-up tag $t_v$}
	\label{initgaug}
	\begin{algorithmic}[1]
		\State $G_{aug} \leftarrow G$
		\State $numClasses_G \leftarrow 1$
		\State $n \leftarrow$ number of nodes in $G$
		\State $reps[1 \ldots n] \leftarrow (null,\ldots,null)$
		\ForAll{$v \in G_{aug}$}
		\State $v_\CLASS \leftarrow 1$
		\State $v_\TAG \leftarrow null$
		\If {$reps[numClasses_G] = null$}
		\State $reps[numClasses_G] \leftarrow v$
		\EndIf
		\EndFor
		\State return $G_{aug}$
	\end{algorithmic}
\end{algorithm}

%Next, we define a procedure $\PARTITION$ that will partition all of the nodes in $G_{aug}$ into equivalence classes by appropriately updating the $\CLASS$ variables of nodes to positive integer values. Two nodes are placed in the same class if and only if their current labels are equal and they were in the same class at the start of the procedure. The partitioning process considers one node at a time and compares its label and previous class value to the representative of each existing class. If no match is found, or no classes exist yet, node $v$ becomes the representative node of a new class. Algorithm \ref{partitionpseudo} gives the pseudocode for $\PARTITION$.

Next, we define a procedure $\PARTITION$ that will partition all of the nodes in $G_{aug}$ into equivalence classes by appropriately updating the $\CLASS$ variables to positive integer values. Two nodes are placed in the same class if and only if they were in the same class at the start of the procedure and their current labels are equal. The partitioning process considers one node at a time and compares its label and previous class value to the representative of each existing class. If no match is found, or no classes exist yet, node $v$ becomes the representative node of a new class. Algorithm \ref{partitionpseudo} gives the pseudocode for $\PARTITION$.

\begin{algorithm}[H]
	\small
	\caption{\PARTITION, input is the augmented configuration $G_{aug}$}
	\label{partitionpseudo}
	\begin{algorithmic}[1]
		\ForAll {$v \in G_{aug}$}\label{storeold}
		\State $oldClass[v] \leftarrow v_{\CLASS}$\Comment{remember each node's class before updating}
		\EndFor\label{storeoldend}
		\ForAll {$v \in G_{aug}$} \label{refineloop}
		\State $\textrm{assignedToClass} \leftarrow false$\label{startrefine}
		%	\State $j \leftarrow 1$
		\For {$k = 1,\ldots,numClasses_G$}\Comment{compare $v$ to existing class reps}\label{checkloopstart}
		\If {$(oldClass[v] = oldClass[reps[k]])$ and $(v_\TAG = reps[k]_\TAG)$}\label{checkmatch}
		\State $v_\CLASS \leftarrow k$\label{assignoldclass}
		\State $\textrm{assignedToClass} \leftarrow true$
		\EndIf
		%	\State $j \leftarrow j+1$
		\EndFor\label{checkloopend}
		\If {$(\textrm{assignedToClass} = false)$}\Comment{create a new class with $v$ as its representative}
		\State $numClasses_G \leftarrow numClasses_G + 1$\label{incclasscount}
		\State $v_\CLASS \leftarrow numClasses_G$\label{assignnewclass}
		\State $reps[numClasses_G] \leftarrow v$
		\EndIf\label{endrefine}
		\EndFor\label{refineloopend}
	\end{algorithmic}
\end{algorithm}

Next, we define a procedure \partitioner that sets the values of the node labels according to the phase that is currently being simulated, and then calls \PARTITION to update the equivalence classes based on the new labels. At a high level, the idea is to record $v$'s history during the phase and succinctly store it in $v$'s label, so that the call to \PARTITION updates the equivalence classes based on node histories. More concretely, to set the label at an arbitrary node $v$, \partitioner first considers each neighbour $w$ of $v$, creates a tuple $(w_{\CLASS}, \sigma+1+t_w - t_v)$, and, if $v_{\CLASS} \neq w_{\CLASS}$ or $t_w \neq t_v$, adds the tuple to a list $N_v$. In fact, we will store each tuple $(a,b)$ as a triple $(a,b,c)$ in $N_v$, where $c=1$ if tuple $(a,b)$ is added exactly once, and $c=*$ otherwise. At a high level, for each neighbour $w$ of $v$, the value $w_{\CLASS}$ represents in which transmission block node $w$ transmits, and the value $\sigma+1+t_w-t_v$ represents $v$'s local round within the block that $w$ transmits. Taken together, the triples of $N_v$ record all of the non-silent rounds in $v$'s history for the current phase, and whether exactly one or more than one neighbour of $v$ transmitted in each such round. A tuple is excluded from $N_v$ when $w_{\CLASS} = v_{\CLASS}$ and $t_w = t_v$ since this represents the case where $v$ and $w$ will transmit at the same time in the current phase, so $v$ would not receive $w$'s transmission nor detect a collision. After completing the construction of $N_v$, \partitioner sets $v$'s label by concatenating together the triples contained in $N_v$. Further, the concatenated triples appear in $v$'s label in increasing order according to the fixed ordering $\prec_{hist}$ given in Definition \ref{ordering} below. This ordering ensures that, when \PARTITION updates the equivalence classes based on node labels, two nodes with equal histories are placed in the same class regardless of the order in which the tuples were added to $N_v$. Algorithm \ref{partitionerpseudo} gives the pseudocode for \partitioner.

\begin{definition}\label{ordering}
	Let $\prec_{hist}$ be the ordering on $\N \times \N \times \{1,*\}$ defined as follows: $(a,b,c) \prec (a',b',c')$ if $a < a'$, or $a=a'$ and $b < b'$, or $a=a'$ and $b=b'$ and $c=1$. 
\end{definition}

\begin{algorithm}[H]
	\small
	\caption{\partitioner, input is the augmented configuration $G_{aug}$}
	\label{partitionerpseudo}
	\begin{algorithmic}[1]
		\ForAll{$v \in G_{aug}$}\label{partvloop}
			\State $N_v \leftarrow$ empty list\label{partvfirst}
			\ForAll{$w$ adjacent to $v$ in $G_{aug}$}\label{neighbourLoop}
				\If{$(w_{\CLASS} \neq v_{\CLASS})$ or $(t_w \neq t_v)$}\label{validTriple}
					\State $newTuple \leftarrow true$
					\ForAll{$(a,b,c) \in N_v$}
						\If{$(a = w_{\CLASS})$ and $(b=\sigma + 1+t_w - t_v)$}\label{foundDuplicate}
							\State $newTuple \leftarrow false$
							\State replace $(a,b,c)$ with $(a,b,*)$ in $N_v$\label{replaceTriple}
						\EndIf
					\EndFor
					\If{$newTuple = true$}
						\State append $(w_{\CLASS}, \sigma + 1+t_w - t_v,1)$ to $N_v$\label{appendTriple}
					\EndIf
				\EndIf
			\EndFor
			\State Sort $N_v$ according to $\prec_{hist}$\label{sortingline}
			\State $v_{\TAG} \leftarrow null$
			\For{$x = 0,\ldots,|N_v|-1$}
				\State $v_{\TAG} \leftarrow v_{\TAG} \cdot N_v[x]$
			\EndFor\label{partvlast}
		\EndFor\label{endpartvloop}

		\State $\PARTITION(G_{aug})$\label{refineline}
		
	\end{algorithmic}
\end{algorithm}

 Finally, we describe the \classifier algorithm, which is designed to output ``Yes" if leader election can be solved on the input configuration $G$, and outputs ``No" otherwise. At a high level, the algorithm starts by initializing the augmented version of configuration $G$, and then executes \partitioner repeatedly. If the nodes are eventually partitioned such that there is an equivalence class containing exactly one node, then the algorithm outputs ``Yes". If no such equivalence class exists, and there is a call to \partitioner such that the partition is the same before and after the call, then the algorithm outputs ``No". Algorithm \ref{classifierpseudo} gives the pseudocode of \classifier.
 
 \begin{algorithm}[H]
 	\small
 	\caption{\classifier, input is configuration $G$, each node $v \in G$ is labeled with its global wake-up tag $t_v$}
 	\label{classifierpseudo}
 	\begin{algorithmic}[1]
 		\State $G_{aug} \leftarrow \initgaug(G)$\label{initline}
 		\For{$i \leftarrow 1,\ldots,\lceil n/2 \rceil$}\label{bigfor}
 		\State $oldClassCount \leftarrow numClasses_G$
 		\State $\partitioner(G_{aug})$
 		\If{$\exists m \in \{1,\ldots,numClasses_G\}$ such that exactly one node $v \in G_{aug}$ has $v_{\CLASS} = m$}\label{yescondition}
 		\State exit and output ``Yes"
 		\EndIf
 		\If{$(numClasses_G =  oldClassCount)$}\label{nocondition}
 		\State exit and output ``No"\label{outputno}
 		\EndIf
 		\EndFor
 		
 	\end{algorithmic}
 \end{algorithm}
 %We will prove that these properties are related to what any algorithm $\cA$ executed by the nodes of $G$ can accomplish.
 
To aid in the analysis of \classifier, we define the following notation. For any $i \in \{1,\ldots,\lceil n/2 \rceil\}$, we refer to the execution of  $\partitioner(G_{aug})$ in the $i^{th}$ iteration of the {\bf for} loop in the execution of \classifier as \emph{iteration $i$ of} \classifier. We refer to the execution of \initgaug as iteration 0 of \classifier. For any $i \in \{1,\ldots,\lceil n/2 \rceil\}$ and any node $v \in G_{aug}$, denote by $v_{\CLASS,i}$ and $v_{\TAG,i}$ the values of $v_{\CLASS}$ and $ v_{\TAG}$, respectively, at the end of iteration $i-1$ of \classifier. Similarly, denote by $numClasses_{G,i}$ and $reps_i$ the value of $numClasses_{G}$ and $reps$, respectively, at the end of iteration $i-1$ of \classifier.

\subsection{Time Complexity of \classifier}
The key to the analysis is to notice that each call to the procedure \PARTITION results in a refinement of the node partition. This is because, according to the condition on line \ref{checkmatch} of \PARTITION, two nodes are assigned to the same class $k$ only if they were in the same class (as the representative of class $k$) immediately before \PARTITION was called. It follows that the number of equivalence classes can only increase. Moreover, as each equivalence class contains at least one node, there cannot be more than $n$ equivalence classes.

\begin{observation}
	For any $v,w$ in $G$ and any $j \geq 1$, if $v_{\CLASS,j} \neq w_{\CLASS,j}$, then $v_{\CLASS,j'} \neq w_{\CLASS,j'}$ for all $j' > j$.
\end{observation}

\begin{corollary}\label{numClassesNonDecreasing}
$1 \leq numClasses_{G,1} \leq \cdots \leq numClasses_{G,\lceil n/2 \rceil+1} \leq n$
\end{corollary}

Using the above fact, we show that \classifier will eventually output `Yes' or `No' (and terminate) in one of its $\lceil n/2 \rceil$ iterations.

\begin{lemma}\label{ClassifierWillTerminate}
	There exists an iteration $i \in \{1,\ldots,\lceil n/2 \rceil\}$ of the {\bf for} loop in \classifier such that either the condition at line \ref{yescondition} or \ref{nocondition} evaluates to true, and then \classifier will terminate.
\end{lemma} 
\begin{proof}
	Assume that the condition at line \ref{nocondition} of \classifier evaluates to false after all iterations $i \in \{1,\ldots,\lceil n/2 \rceil\}$ of \classifier. Corollary \ref{numClassesNonDecreasing} implies that $1 \leq numClasses_{G,1} < \cdots < numClasses_{G,\lceil n/2 \rceil+1}$. It follows that $numClasses_{G,\lceil n/2 \rceil+1} \geq \lceil n/2 \rceil+1$. In other words, at the end of iteration $\lceil n/2 \rceil$ of \classifier, the number of equivalence classes is strictly greater than $n/2$, which means that the average class size is strictly less than 2. It follows that there is at least one class with size exactly 1 immediately after $\partitioner(G_{aug})$ is executed for the $\lceil n/2 \rceil^{th}$ time, and the condition at line \ref{yescondition} evaluates to true.
\end{proof}

We showed above that there are at most $\lceil n/2 \rceil$ iterations of \classifier. Each iteration assigns a node label to each node, and then updates the node partition by check the equality of pairs of labels. We can show that each such iteration takes at most $O(n^2\Delta)$ steps, where $\Delta$ is the maximum degree of the nodes in $G$. This gives an overall running time of $O(n^3\Delta)$, which is $O(n^4)$ in the worst case.

\begin{lemma}\label{ClassifierComplexity}
	The worst-case time complexity of \classifier on an input configuration $G$ is $O(n^3\Delta)$, where $\Delta$ is the maximum degree of $G$.
\end{lemma}
\begin{proof}
	We assume that any comparison of two $O(\log n)$-bit or $O(\log \sigma)$-bit words takes constant time. In particular, note that comparing two triples in $N_v$ in \partitioner takes constant time. Let $\Delta$ be the maximum degree of nodes in $G$. We determine an upper bound on the worst-case time complexity of \partitioner. 
	
	First, consider an arbitrary $v \in G_{aug}$ in the execution of the {\bf for} loop at lines \ref{partvloop}-\ref{endpartvloop}. We determine an upper bound on the number of triples in $N_v$. Each triple is added at line \ref{appendTriple} of \partitioner, and this happens at most once per iteration of the {\bf for} loop at line \ref{neighbourLoop}. As this loop iterates through all neighbours of $v$, it follows that $N_v$ contains at most $\Delta$ triples. The overall complexity of determining $N_v$ is $\Delta^2$: for each triple added to $N_v$, it was compared with all previously added triples, of which there are at most $\Delta$. The complexity of sorting $N_v$ at line \ref{sortingline} is $O(\Delta\log \Delta)$. Appending together the triples in $N_v$ to form $v_{\TAG}$ takes $O(\Delta)$ time. Thus, the worst-case complexity of lines \ref{partvfirst}-\ref{partvlast} is dominated by $O(\Delta^2)$. As this is repeated for each $v \in G_{aug}$, the worst-case time complexity of the {\bf for} loop at lines \ref{partvloop}-\ref{endpartvloop} is $O(n\Delta^2)$.
	
	Next, consider the execution of $\PARTITION(G_{aug})$ at line \ref{refineline} of \partitioner. At line \ref{checkmatch} in \PARTITION, comparing $v_{\TAG}$ with $reps[k]_{\TAG}$ takes $O(\Delta)$ time: each label consists of at most $\Delta$ triples, the labels can be scanned from left-to-right as the triples are in sorted order, and comparing two triples takes constant time. This is performed $numClasses_G \leq n$ times inside the {\bf for} loop at lines \ref{checkloopstart}-\ref{checkloopend} in \PARTITION. Thus, the complexity of assigning a class number to an arbitrary node $v$ (lines \ref{startrefine}-\ref{endrefine}) is dominated by $O(n\Delta)$. This is repeated $n$ times in the {\bf for} loop at line \ref{refineloop}. So the running time of lines \ref{refineloop}-\ref{refineloopend} is $O(n^2\Delta)$. This dominates the $O(n)$ steps needed for lines \ref{storeold}-\ref{storeoldend}. Thus, the number of steps taken by the execution of $\PARTITION(G_{aug})$ at line \ref{refineline} of \partitioner is bounded above by $O(n^2\Delta)$.
	
	From the above discussion, we see that an execution of \partitioner has worst-case time complexity $O(n^2\Delta)$. Finally, since by Lemma \ref{ClassifierWillTerminate}, \partitioner is executed at most $\lceil n/2 \rceil$ times by \classifier, we get that the overall worst-case time complexity of \classifier is $O(n^3\Delta)$.

\end{proof}

%One important corollary of Observation \ref{partitionbase} is the fact that, as soon as two nodes are placed in different equivalence classes, they will never be placed in the same equivalence class in any subsequent iteration of \partitioner. This highlights the important property that each call to \partitioner only refines the node partition, and, in particular, means that the number of equivalence classes is non-decreasing.
%\begin{corollary}\label{onlyrefine}
%	For any $v,w \in G_{aug}$ and any $i \geq 0$, if $v_{\CLASS,i} \neq w_{\CLASS,i}$, then $v_{\CLASS,i+1} \neq w_{\CLASS,i+1}$.
%\end{corollary}
%The previous observation  In particular, equivalence classes can only be split into smaller classes, two different classes cannot merge into one, and the number of equivalence classes is non-decreasing.

%\begin{corollary}\label{classesnondecreasing}
%	For any $i \geq 0$, $numClasses_{G,i} \leq numClasses_{G,i+1}$.
%\end{corollary}
%
%

%% file: characterization.tex
In this section, we show that \classifier correctly identifies whether or not a configuration $G$ is feasible. First, in Section \ref{canonical}, we give a distributed implementation (a DRIP) of the iterations of \classifier. We call this the \emph{canonical DRIP} for configuration $G$ since the outcome of its execution accurately predicts whether or not there exists \emph{any} DRIP that can be used to solve leader election in $G$. In Section \ref{canonicalProperties}, we prove some important properties about the canonical DRIP. We then proceed to prove the correctness of \classifier. In Section \ref{yesInstances}, we prove that if \classifier outputs ``Yes" on input $G$, then the canonical DRIP can be used to solve leader election in $G$, and thus $G$ is feasible. Conversely, in Section \ref{noInstances}, we prove that if $G$ is feasible, then the canonical DRIP can be used to solve leader election in $G$, and \classifier outputs ``Yes" on input $G$.

\subsubsection{The Canonical DRIP}\label{canonical}

In order to implement \classifier in a distributed way for a given configuration $G$, we would like each node to be able to independently determine in which equivalence class of the partition it belongs to at the start of each phase. However, as the nodes of $G$ are anonymous, we must install an identical algorithm at each node. So, in order to describe the canonical DRIP for configuration $G$, we first create a sequence of lists that is defined with respect to the execution of \classifier on $G$. The high-level idea is that, for each $j \geq 1$, the $j^{th}$ list consists of a list of histories of class representatives chosen by iteration $j-1$ of \classifier. The same sequence of lists is hard-coded into the algorithm installed at each node, and at the start of each phase $P_j$, each node $v$ will compare its history to the items in list $\mathcal{L}_j$ to determine which equivalence class it was assigned to by \classifier. It will use this class number to determine when it will transmit during phase $P_j$.

We now give a detailed construction of each list $\mathcal{L}_j$. If \classifier terminates after iteration $j-1$, then $\mathcal{L}_j$ consists of a single item: the string ``terminate". Otherwise, each item in list $\mathcal{L}_j$ corresponds to a single class representative, whose history is encoded as a tuple $(oldClass, label)$. For $j \geq 2$, the value of $oldClass$ represents the class number that the representative belonged to at the start of the phase $P_{j-1}$, and the value of $label$ corresponds to the history of the representative during phase $P_{j-1}$. For $j=1$, we pick a tuple that reflects the fact that, at initialization, all nodes are in the same class and have no history. More specifically:
\begin{itemize}
	\item $\mathcal{L}_1$ consists of one item: the tuple $(1,null)$.
	%Every node's history will match this tuple at the start of phase $P_0$: as there was no previous phase, no node has a class number yet (the smallest class number is 1) and there is no previous history. This agrees with the fact that all nodes are initially placed by \classifier into a single equivalence class. 
	\item For each $j \geq 2$,
\begin{itemize}
	\item if $numClasses_{G,j} = numClasses_{G,j-1}$, or $\exists m \in \{1,\ldots,numClasses_{G,j}\}$ such that exactly one node $v \in G_{aug}$ has $v_{\CLASS,j} = m$, then the list $\mathcal{L}_j$ consists of one item, defined as $\mathcal{L}_j[1] = ``terminate"$.
	\item Otherwise, $\mathcal{L}_j$ is a list consisting of $numClasses_{G,j}$ tuples. For each $k \in \{1,\ldots,numClasses_{G,j}\}$, define the $k^{th}$ tuple in the list as $\mathcal{L}_j[k] = (reps_{j}[k]_{\CLASS,j-1},reps_j[k]_{\TAG,j})$. To clarify, we take the representative of class $k$ at the end of phase $P_{j-1}$, and write the number of the class it was in at the end of phase $P_{j-2}$ into the first entry of the tuple, and write the label it was assigned during phase $P_{j-1}$ into the second entry of the tuple.\\ Recall that the label $reps_j[k]_{\TAG,j}$ is a concatenated sequence of triples of the form $(a,b,c)$.
	%In other words, the $k^{th}$ item in list $\mathcal{L}_j$ consists of the class number and the label assigned to the representative of class $k$ in iteration $j-1$ of \classifier. 

\end{itemize}
\end{itemize}
%Let $z$ be the smallest positive integer such that $\mathcal{L}_z$ consists of the string ``terminate". 
%Define a table $T$ with $z+1$ entries: for $j \in \{0,\ldots,z\}$, define $T[j] = \mathcal{L}_j$. 

We now provide the specification of the canonical DRIP \canonDRIP. For a given configuration $G$, the canonical DRIP for $G$ is defined locally at each node $v$ as a sequence of phases that starts when node $v$ wakes up. Define $r_0$ to be local round 0, and let phase $P_0$ consist of local round 0. For each $j \geq 1$, the phase $P_j$ and local round $r_{j}$ in which phase $P_j$ ends are defined inductively, as follows. At each node $v$, phase $P_{j}$ starts in local round $r_{j-1}+1$. There are two possible cases:
\begin{itemize}
	\item If $\mathcal{L}_{j}[1] = ``terminate"$ then node $v$ terminates in local round $r_{j-1}+1$ and $r_{j}$ is defined to be $r_{j-1}+1$.
	\item Otherwise, phase $P_j$ consists of $numClasses_{G,j}$ consecutive transmission blocks, each block consisting of exactly $2\sigma+1$ rounds, followed by $\sigma$ rounds in which node $v$ listens. The local round $r_{j}$ in which phase $P_{j}$ ends is set to $r_{j-1} + numClasses_{G,j}\cdot (2\sigma+1) + \sigma$. Node $v$ can deduce the value of $numClasses_{G,j}$ by looking at the number of items in the list $\mathcal{L}_j$. 
	
	During phase $P_j$, node $v$ transmits during exactly one transmission block. It maintains a variable $tBlock$ to store the block number in which it should transmit, and this value gets re-calculated at the start of each phase. The initial value of $tBlock$ before the start of phase $P_1$ is 1. Its transmission in phase $P_j$ occurs in the $(\sigma+1)^{th}$ round of block $tBlock$, and the transmitted message is a `1'. Node $v$ listens in all other rounds of the phase. 
	
	To calculate the value of $tBlock$ at the start of a phase $P_j$, node $v$ finds the value of $k$ such that the $k^{th}$ entry in list $\mathcal{L}_j$ ``matches" $v$'s history during the rounds of phase $P_{j-1}$. In particular, for each $k \in \{1,\ldots,numClasses_{G,j}\}$, node $v$ performs the following comparison between $\mathcal{L}_j[k] = (oldClass_k,label_k)$ and its history during phase $P_{j-1}$: 
	
	\begin{enumerate}
		\item Check that $oldClass_k = tBlock$. Note that we are currently calculating the value of $tBlock$ of phase $P_j$, so this comparison is performed with the old value of $tBlock$.
	
		\item If $j \geq 2$, for each round $t$ in phase $P_{j-1}$ such that $t = r_{j-2}+(a-1)(2\sigma+1)+b$ for some $a \in \{1,\ldots,numClasses_{G,j-1}\}$ and $b \in \{1,\ldots,2\sigma+1\}$:
		
			\begin{itemize}
				\item If $\mathcal{H}_{v,\canonDRIP}[t] = (\textrm{`1'})$, check that there exists a triple $(a,b,1)$ in $label_k$.
				\item If $\mathcal{H}_{v,\canonDRIP}[t] = (*)$, check that there exists a triple $(a,b,*)$ in $label_k$.
				\item If $\mathcal{H}_{v,\canonDRIP}[t] = (\emptyset)$, check that there exists no triple $(a,b,1)$ or $(a,b,*)$ in $label_k$.
			\end{itemize}
\end{enumerate}
If any of the checks returns false, $v$ immediately aborts the comparison, increments $k$, and tries again. When all checks return true for some value of $k$, node $v$ sets $tBlock = k$.
\end{itemize} 

\subsubsection{Properties of \canonDRIP}\label{canonicalProperties}
In this section, we prove some properties about the canonical DRIP that will be important in the proof of correctness of \classifier. The first property is that \canonDRIP is a patient DRIP, which implies that all nodes wake up spontaneously in the global round equal to their wakeup tag.

\begin{lemma}\label{canonPatient}
	In the execution of \canonDRIP by all nodes in $G$, no node transmits in any of the global rounds $0,\ldots,\sigma$. Equivalently, each node in $G$ wakes up spontaneously in the execution of \canonDRIP.
\end{lemma}
\begin{proof}
To prove the claim, we proceed by induction on the global round number $r$. In the base case, $r=0$, all nodes with wakeup tag greater than 0 are not awake, and those with wakeup tag equal to 0 start executing their local algorithm in global round 1. Thus, no node transmits in global round 0. As induction hypothesis, assume that no node transmits in global rounds $0,\ldots,k-1$ for some $k \in \{1,\ldots,\sigma\}$. We consider the behaviour of an arbitrary node $v$ in global round $k$. By the induction hypothesis, $v$ does not receive any messages in global rounds $0,\ldots,k-1$. We consider two cases:
\begin{itemize}
	\item if $v$ has wakeup tag greater than or equal to $k$, then the fact that $v$ does not receive any messages in global rounds $0,\ldots,k-1$ implies that $v$ is not awake before global round $k$. Therefore, $v$ does not start executing its local algorithm until global round $k+1$ or later, so $v$ does not transmit in global round $k$.
	\item if $v$ has wakeup tag less than $k$, then, from the description of \canonDRIP, the earliest local round in which $v$ might transmit is the $(\sigma+1)^{th}$ round of the first transmission block of phase $P_1$. In particular, it does not transmit before its local round $\sigma+1$. As its local clock value is always bounded above by the global clock value, and we assumed that $k \leq \sigma$, we conclude that $v$ does not transmit during global round $k$.
\end{itemize}
In all cases, $v$ does not transmit in round $k$, so it follows that no node transmits in global rounds $0,\ldots,k$, which completes the inductive step.
\end{proof}

An important feature of the canonical DRIP's design is that, in each phase $P_j$, every node transmits exactly once. This is because each node is placed in some equivalence class by \classifier, each equivalence class is assigned a transmission block within the phase, and each node in the equivalence class transmits in its local round $\sigma+1$ within that block. Due to offsets in wakeup times, local round $\sigma+1$ might correspond to a different local rounds at a different nodes, however the transmission block and phase number will be the same. As a result, in each phase, each node will have a chance to receive a transmission from each of its neighbours. Further, the round in which each transmission occurs relative to the start of a transmission block reveals the relative offset of the transmitting node and a neighbour. These observations are formalized in the following result.

\begin{lemma}\label{lem:transmitConversions}
	For any configuration $G$, consider the execution of the canonical DRIP for $G$, and consider any $j \geq 1$ such that $\mathcal{L}_j \neq ``terminate"$. For any $v$ in $G$, for any $h \in \{1,\ldots,2\sigma+1\}$ and for any $k \in \{1,\ldots,numClasses_{G,j}\}$, consider the $h$'th round of the $k$'th transmission block in phase $P_j$ of node $v$'s execution of the canonical DRIP, i.e., $v$'s local round $r_{j-1} + (k-1)(2\sigma+1)+h$. A neighbour $\hat{v}$ of $v$ transmits in this round if and only if 
	\begin{itemize}
		\item $h = \sigma+1+t_{\hat{v}} - t_v$, and,
		\item $\hat{v}$ transmits in transmission block $k$ of phase $P_j$ of its execution of \canonDRIP. %$\canonHist{\hat{v}}{r_{j-1}} = \canonHist{reps_j[k]}{r_{j-1}}$ and 
	\end{itemize}
\end{lemma}
\begin{proof}
	In the execution of the DRIP at node $v$, phase $P_j$ starts in local round $r_{j-1}+1$ and each transmission block consists of $2\sigma+1$ rounds. We conclude that the $h$'th round of the $k$'th transmission block in phase $P_j$ at node $v$ has local round number $r_{j-1} + (k-1)(2\sigma+1)+h$. By Proposition \ref{prop:localconversion}, the local round number at $\hat{v}$ corresponding to this round is $r_{j-1} + (k-1)(2\sigma+1) + h - (t_{\hat{v}} - t_v)$. Since $h \in \{1,\ldots,2\sigma+1\}$ and $|t_{\hat{v}} - t_v| \leq \sigma$ (by the definition of the span $\sigma$) we get that $r_{j-1}+(k-1)(2\sigma+1) + h - (t_{\hat{v}} - t_v) \in \{r_{j-1} + (k-1)(2\sigma+1)-\sigma+1,\ldots,r_{j-1}+(k-1)(2\sigma+1)+3\sigma+1\}$. It follows that the local round $r_{j-1}+(k-1)(2\sigma+1) + h - (t_{\hat{v}} - t_v)$ at $\hat{v}$ falls into one of the following three ranges:
	\begin{enumerate}
		\item$\{r_{j-1}+(k-1)(2\sigma+1)-\sigma+1,\ldots,r_{j-1}+(k-1)(2\sigma+1)\}$,
		\item $\{r_{j-1}+(k-1)(2\sigma+1)+1,\ldots,r_{j-1}+(k-1)(2\sigma+1)+2\sigma+1\}$, or 
		\item $\{r_{j-1}+(k-1)(2\sigma+1)+2\sigma+2,\ldots,r_{j-1}+(k-1)(2\sigma+1)+3\sigma+1\}$. 
	\end{enumerate}
	The local rounds in the first range are the $\sigma$ rounds preceding the start of the $k$'th transmission block, which means that they are either the last $\sigma$ rounds of the previous transmission block, or the last $\sigma$ rounds of the previous phase. In both cases, by the definition of the DRIP, $\hat{v}$ does not transmit in these rounds. The local rounds in the third range are the $\sigma$ rounds after the end of the $k$'th transmission block, which means that they are either the first $\sigma$ rounds of the next transmission block, or the $\sigma$ rounds at the end of phase $P_j$. In both cases, by the definition of the DRIP, $\hat{v}$ does not transmit in these rounds. The local rounds in the second range are the $2\sigma+1$ rounds of the $k$'th transmission block in phase $P_j$ in the execution by $\hat{v}$. 
	%By Lemma \ref{kClassBlock}, node $\hat{v}$ transmits in the $k$'th transmission block of phase $P_{j}$ if and only if $\hat{v}_{\CLASS,j} = k$. Lemma \ref{diffClassesHistories} implies that $\canonHist{\hat{v}}{r_{j-1}} = \canonHist{reps_j[k]}{r_{j-1}}$. 
	Moreover, node $\hat{v}$ may only transmit in the $(\sigma+1)$'th round of the transmission block (which corresponds to local round $r_{j-1}+(k-1)(2\sigma+1) + \sigma + 1$). Setting $r_{j-1} + (k-1)(2\sigma+1) + h - (t_{\hat{v}} - t_v) = r_{j-1}+(k-1)(2\sigma+1) + \sigma + 1$, we see that this is the case if and only if $h = \sigma + 1+t_{\hat{v}} - t_v$, as desired.
\end{proof}

An important relationship between \classifier and the canonical DRIP \canonDRIP is that all nodes that are placed in the same equivalence class by an iteration $j$ of \classifier have the same history up to the end of phase $P_j$ in the execution of \canonDRIP. This is because the label $v_{\TAG}$ assigned to each node $v$ in one iteration of \classifier accurately ``encodes" the history of $v$'s execution in the corresponding phase of the canonical DRIP, as we demonstrate in the next result. Further, we verify that a node $v$'s equivalence class number in \classifier is equal to the transmission block number in which it transmits in \canonDRIP.

\begin{lemma}\label{relateClassifierCanonical}
	For every $j \geq 1$ such that $\mathcal{L}_j[1] \neq ``terminate"$, the following two statements hold for each node $v$ in $G$:
	\begin{enumerate}[label=(\arabic*)]
		
		\item If $j \geq 2$, then for each local round $t$ in phase $P_{j-1}$ such that $t$ is the $b^{th}$ round within transmission block $a$ for some $a \in \{1,\ldots,numClasses_{G,j-1}\}$ and $b \in \{1,\ldots,2\sigma+1\}$: \label{histLabel}
		\begin{itemize}
			\item $\mathcal{H}_{v,\canonDRIP}[t] = (\textrm{`1'})$ if and only if there exists a triple $(a,b,1)$ in $v_{\TAG,j}$.
			\item $\mathcal{H}_{v,\canonDRIP}[t] = (*)$ if and only if there exists a triple $(a,b,*)$ in $v_{\TAG,j}$.
			%\item If $\mathcal{H}_{v,\canonDRIP}[t] = (\emptyset)$, then there exists no triple $(a,b,1)$ or $(a,b,*)$ in $v_{\TAG,j+1}$.
		\end{itemize}
		\item Node $v$ transmits in transmission block $k$ of phase $P_j$ in its execution of the canonical DRIP \canonDRIP if and only if $v_{\CLASS,j} = k$ in the execution of \classifier.\label{kClassBlocks}
	\end{enumerate}

\end{lemma}
\begin{proof}
	The proof proceeds by induction on the value of $j$. 
	
	\underline{Base Case:} For the base case, consider $j=1$. 
	
	Statement \ref{histLabel} is vacuously true as $j < 2$.
	
	To prove statement \ref{kClassBlocks}, it is sufficient to prove that, for each $v$ in $G$, we have $v_{\CLASS,1}=1$ and $v$ transmits in transmission block 1 of phase $P_1$. In $\initgaug$, each node $v$ in $G$ is assigned to equivalence class 1, i.e., $v_{\CLASS,1} = 1$. Moreover, in $\initgaug$, the value of $numClasses$ is set to 1, i.e., $numClasses_{G,1} = 1$. Therefore, by the definition of \canonDRIP, phase $P_1$ has exactly one transmission block. We verify that all nodes transmit in this transmission block. For an arbitrary node $v$, its initial value for $tBlock$ before the first phase is 1. By the definition of \canonDRIP, when $j=1$, each node determines in which block it will transmit by comparing its initial value of $tBlock$ to the first entry of the tuple $\mathcal{L}_1[1] = (1,null)$. Thus, $v$ will set $tBlock=1$ and transmit during transmission block 1, as required.

	\underline{Induction Hypothesis:} Assume that statements \ref{histLabel} and \ref{kClassBlocks} hold for some $j \geq 1$ such that $\mathcal{L}_j[1] \neq ``terminate"$.

	\underline{Inductive Step:} Suppose that $\mathcal{L}_{j+1}[1] \neq ``terminate"$.

	To prove \ref{histLabel}, consider any local round $t$ in phase $P_{j}$ such that $t$ is the $b^{th}$ round in transmission block $a$ for some $a \in \{1,\ldots,numClasses_{G,j}\}$ and $b \in \{1,\ldots,2\sigma+1\}$.
	
	First, from the definition of \canonDRIP, observe that $v$ listens in its local round $t$, i.e., the $b^{th}$ round within transmission block $a$, if and only if $b \neq \sigma+1$ or $v$ doesn't transmit in transmission block $a$ of phase $P_{j}$ at all. By statement \ref{kClassBlocks} of the induction hypothesis, $v$ doesn't transmit in transmission block $a$ of phase $P_{j}$ if and only if $v_{\CLASS,j} \neq a$.
	
	Also, by Lemma \ref{lem:transmitConversions}, a node $w$ transmits in the $b^{th}$ round within transmission block $a$ of phase $P_{j}$ at node $v$ if and only if $b = \sigma+1+t_w-t_v$ and $w$ transmits in transmission block $a$ of phase $P_{j}$ of its execution of \canonDRIP. By statement \ref{kClassBlocks} of the induction hypothesis, $w$ transmits in transmission block $a$ of phase $P_{j}$ of its execution of \canonDRIP if and only if $w_{\CLASS,j} = a$.
	
	Combining the above, we get that $v$ listens in round $t$ and $w$ transmits in this round if and only if $b = \sigma+1+t_w-t_v$ and $w_{\CLASS,j} = a$ and at least one of $b \neq \sigma+1$ or $v_{\CLASS,j} \neq a$ is true. In particular, $v$ listens in round $t$ and $w$ transmits in this round if and only if all of the following hold:
	\begin{enumerate}[label=(\roman*)]
		\item $a = w_{\CLASS,j}$, \label{classA}
		\item $b = \sigma+1+t_w-t_v$, and, \label{roundB}
		\item $v_{\CLASS,j} \neq w_{\CLASS,j}$ or $t_w \neq t_v$.\label{vReceives}
	\end{enumerate}
	
	\paragraph{}
	We can now prove the two statements of \ref{histLabel}:
	\begin{itemize}
		\item From the definition of \partitioner, note that a triple $(a,b,1)$ appears in $v_{\TAG}$ if and only if both of the following hold:
		
		\begin{itemize}
			\item there is a neighbour $w$ of $v$ (i.e., an iteration of the {\bf for} loop at line \ref{neighbourLoop}) such that:
			\begin{itemize}
				\item the condition at line \ref{validTriple} is satisfied, i.e., if and only if $(w_{\CLASS} \neq v_{\CLASS})$ or $(t_w \neq t_v)$, and,
				\item the triple $(a,b,1)$ is added at line \ref{appendTriple}, i.e., $a = w_{\CLASS}$ and $b=\sigma+1+t_w-t_v$,
			\end{itemize}
			\item there is no other neighbour $\hat{w}$ of $v$ (i.e., no other iteration of the {\bf for} loop at line \ref{neighbourLoop}), such that line \ref{replaceTriple} is executed, i.e., no other neighbour $\hat{w}$ of $v$ such that the conditions on lines \ref{validTriple} and \ref{foundDuplicate} both hold. In other words, there is no node $\hat{w} \neq w$ adjacent to $v$ such that $(a = \hat{w}_{\CLASS})$ and $(b = \sigma + 1+t_{\hat{w}} - t_v)$ and at least one of $(w_{\CLASS} \neq v_{\CLASS})$ or $(t_w \neq t_v)$ holds.
		\end{itemize}
		
		In particular, a triple $(a,b,1)$ appears in $v_{\TAG}$ if and only if there is exactly one neighbour $w$ of $v$ such that conditions \ref{classA} - \ref{vReceives} all hold. We proved above that these conditions all hold for a node $w$ if and only $v$ listens in its local round $t$ and $w$ transmits during this round. Therefore, a triple $(a,b,1)$ appears in $v_{\TAG}$ if and only if $v$ listens in local round $t$ and has exactly one neighbour $w$ that transmits in this round. According to \canonDRIP, each transmission consists of the string `1', so we conclude that $(a,b,1)$ appears in $v_{\TAG}$ if and only if $\mathcal{H}_{v,\canonDRIP}[t] = (\textrm{`1'})$, as desired.

		\item From the definition of \partitioner, note that a triple $(a,b,*)$ appears in $v_{\TAG}$ if and only if both of the following hold:
		
		\begin{itemize}
			\item there is a neighbour $w$ of $v$ (i.e., an iteration of the {\bf for} loop at line \ref{neighbourLoop}) such that:
			\begin{itemize}
				\item the condition at line \ref{validTriple} is satisfied, i.e., if and only if $(w_{\CLASS} \neq v_{\CLASS})$ or $(t_w \neq t_v)$, and,
				\item the triple $(a,b,1)$ is added at line \ref{appendTriple}, i.e., $a = w_{\CLASS}$ and $b=\sigma+1+t_w-t_v$,
			\end{itemize}
			\item there is at least one other neighbour $\hat{w}$ of $v$ (i.e., at least one other iteration of the {\bf for} loop at line \ref{neighbourLoop}), such that line \ref{replaceTriple} is executed, i.e., at least one other neighbour $\hat{w}$ of $v$ such that the conditions on lines \ref{validTriple} and \ref{foundDuplicate} both hold. In other words, there is at least one node $\hat{w} \neq w$ adjacent to $v$ such that $(a = \hat{w}_{\CLASS})$ and $(b = \sigma + 1+t_{\hat{w}} - t_v)$ and at least one of $(w_{\CLASS} \neq v_{\CLASS})$ or $(t_w \neq t_v)$ holds.
		\end{itemize}
		
		In particular, a triple $(a,b,*)$ appears in $v_{\TAG}$ if and only if there are at least two neighbours of $v$ such that conditions \ref{classA} - \ref{vReceives} all hold. We proved above that these conditions all hold for a node $w$ if and only $v$ listens in its local round $t$ and $w$ transmits during this round. Therefore, a triple $(a,b,*)$ appears in $v_{\TAG}$ if and only if $v$ listens in local round $t$ and at least two neighbours transmit in this round, i.e., a collision occurs. We conclude that $(a,b,*)$ appears in $v_{\TAG}$ if and only if $\mathcal{H}_{v,\canonDRIP}[t] = (*)$, as desired.
	\end{itemize}

	To prove \ref{kClassBlocks}, consider any node $v$ in $G$. For any $k \in \{1,\ldots,numClasses_{G,j+1}\}$, denote the tuple stored in $\mathcal{L}_{j+1}[k]$ by $(oldClass_k,label_k)$. From the description of \canonDRIP, node $v$ transmits in transmission block $k \in \{1,\ldots,numClasses_{G,j+1}\}$ of phase $P_{j+1}$ in its execution of \canonDRIP if and only if  both of the following conditions are satisfied:
	\begin{enumerate}[label=(\alph*)]
		\item $oldClass_k$ is the block in which $v$ transmits in phase $P_j$.\label{oldClass}
		\item for each round $t$ in phase $P_{j}$ such that $t = r_{j-1}+(a-1)(2\sigma+1)+b$ for some $a \in \{1,\ldots,numClasses_{G,j}\}$ and $b \in \{1,\ldots,2\sigma+1\}$:\label{sameTag}
		
		\begin{itemize}
			\item If $\mathcal{H}_{v,\canonDRIP}[t] = (\textrm{`1'})$, then there exists a triple $(a,b,1)$ in $label_k$.
			\item If $\mathcal{H}_{v,\canonDRIP}[t] = (*)$, then there exists a triple $(a,b,*)$ in $label_k$.
			\item If $\mathcal{H}_{v,\canonDRIP}[t] = (\emptyset)$, then there exists no triple $(a,b,1)$ or $(a,b,*)$ in $label_k$.
		\end{itemize}
	\end{enumerate}
	By statement \ref{kClassBlocks} of the induction hypothesis, condition \ref{oldClass} is true if and only if $v_{\CLASS,j} = oldClass_k$. By statement \ref{histLabel} proven above, we have $\mathcal{H}_{v,\canonDRIP}[t] = (\textrm{`1'})$ if and only if there exists a triple $(a,b,1)$ in $v_{\TAG,j+1}$, and $\mathcal{H}_{v,\canonDRIP}[t] = (*)$ if and only if there exists a triple $(a,b,*)$ in $v_{\TAG,j+1}$. Thus, condition \ref{sameTag} is true if and only if $v_{\TAG,j+1} = label_k$. By definition, $(oldClass_k,label_k) = \mathcal{L}_{j+1}[k] = (reps_{j}[k]_{\CLASS,j},reps_{j+1}[k]_{\TAG},j+1)$, so we have shown that conditions \ref{oldClass} and \ref{sameTag} are true if and only if $v_{\CLASS,j} = reps_{j}[k]_{\CLASS,j}$ and $v_{\TAG,j+1} = reps_{j+1}[k]_{\TAG,j+1}$. In other words, the two conditions are true if and only if node $v$ is in the same equivalence class as node $reps_{j}[k]$ at the end of iteration $j-1$ of \classifier, and node $v$ is assigned the same label as $reps_{j}[k]$ during phase $P_j$, and this occurs if and only if the two conditions at line \ref{checkmatch} in the execution of \PARTITION are satisfied. The two conditions at line \ref{checkmatch} are satisfied if and only if line \ref{assignoldclass} is executed, i.e., $v_{\CLASS,j+1} = k$, as desired.
\end{proof}

The crucial relationship between \classifier and the canonical DRIP \canonDRIP is the following: two nodes are in different equivalence classes at the end of iteration $j$ of \classifier if and only if the two nodes have different histories after executing phase $P_j$ of the canonical DRIP. This relationship will form the basis of the proof of correctness of \classifier: if \classifier outputs ``Yes" for a configuration $G$, then \classifier creates at least one equivalence class that consists of exactly one node, and therefore this node has a unique history in the execution of \canonDRIP, and can be chosen as leader of $G$; conversely, if $G$ is feasible, we will be able to show that \canonDRIP elects some leader $v$ in $G$, which implies $v$ has a unique history in the execution of \canonDRIP, and therefore \classifier puts $v$ in its own equivalence class, so \classifier outputs ``Yes" for configuration $G$.

\begin{lemma}\label{lem:diffClassesHistory}
	For any two nodes $v,w$ in $G$ and any $j \geq 1$, we have $\canonHist{v}{r_{j-1}} = \canonHist{w}{r_{j-1}}$ if and only if $v_{\CLASS,j} = w_{\CLASS,j}$.
\end{lemma}
\begin{proof}
	The proof proceeds by induction on $j$. For the base case, consider $j=1$. The statement is vacuously true in both directions. First, for all $v,w$ in $G$, note that \initgaug assigns all nodes to equivalence class 1, i.e., $v_{\CLASS,1} = w_{\CLASS,1} = 1$. Conversely, by Lemma \ref{canonPatient}, all nodes wake up spontaneously when executing \canonDRIP, so $\canonHist{v}{r_{0}} = \mathcal{H}_{v,\canonDRIP}[0] = (\emptyset) = \mathcal{H}_{w,\canonDRIP}[0] = \canonHist{w}{r_{0}}$.
	
	As induction hypothesis, assume that for any two nodes $v,w$ in $G$ and some $j \geq 1$, we have $\canonHist{v}{r_{j-1}} = \canonHist{w}{r_{j-1}}$ if and only if $v_{\CLASS,j} = w_{\CLASS,j}$.
	
	For the induction step, consider any two nodes $v,w$ in $G$. 
	
	Nodes $v$ and $w$ are placed in the same equivalence class $k$ by \PARTITION if and only if the two conditions at line \ref{checkmatch} in the execution of \PARTITION are satisfied for $v$ and $w$ for the same value of $k$. In other words, $v_{\CLASS,j+1} = w_{\CLASS,j+1}$ if and only if $v_{\CLASS,j} = w_{\CLASS,j}$ and $v_{\TAG,j+1} = w_{\TAG,j+1}$. By the induction hypothesis, we know that $v_{\CLASS,j} = w_{\CLASS,j}$ if and only if $\canonHist{v}{r_{j-1}} = \canonHist{w}{r_{j-1}}$. 
	
	Next, by statement \ref{histLabel} of Lemma \ref{relateClassifierCanonical}, we know that for each local round value $t$ such that $t$ is the $b^{th}$ round within transmission block $a$ of phase $P_j$, a triple $(a,b,1)$ contained in a node's label corresponds to a $(\textrm{`1'})$ in entry $t$ of the node's history, and a triple $(a,b,*)$ contained in a node's label corresponds to a $(*)$ in entry $t$ of the node's history. Thus, $v_{\TAG,j+1} = w_{\TAG,j+1}$ if and only if $\mathcal{H}_{v,\canonDRIP}[t] = \mathcal{H}_{w,\canonDRIP}[t]$ for all rounds in phase $P_j$, i.e., $\mathcal{H}_{v,\canonDRIP}[r_{j-1}+1\ldots r_j] = \mathcal{H}_{w,\canonDRIP}[r_{j-1}+1\ldots r_j]$. 
	
	Thus, we have shown that $\canonHist{v}{r_{j}} = \canonHist{w}{r_{j}}$, as required.
\end{proof}

Finally, we give a bound on the time complexity of the canonical DRIP.

\begin{lemma}\label{CanonicalComplexity}
	In the execution of the canonical DRIP \canonDRIP, each node terminates within $O(n^2\sigma)$ rounds.
\end{lemma}
\begin{proof}
	By Lemma \ref{ClassifierWillTerminate} and the definition of \canonDRIP, the number of phases will be at most $\lceil n/2 \rceil$. Each phase consists of at most $n$ transmission blocks (one per equivalence class) and exactly $\sigma$ additional rounds. Each block consists of $2\sigma+1$ rounds.
\end{proof}

%\begin{lemma}
%	Suppose that \classifier exits after some iteration $i$ of \partitioner. Then, for each node $v$ executing the canonical DRIP \canonDRIP:
%	\begin{itemize}
%		\item $v$ terminates its execution of \canonDRIP in the local round immediately after the end of phase $P_i$, and,
%		\item for each $j \leq i$, if $v_{\CLASS,j} = k$, then node $v$ transmits in local round $\sigma+1$ of transmission block $k$ of phase $P_j$. Node $v$ does not transmit in any other round of phase $P_j$.
%	\end{itemize}
%\end{lemma}

\subsubsection{Correctness of \classifier: ``Yes" Instances}\label{yesInstances}
In this section, we prove the correctness of \classifier in cases where it outputs ``Yes" for input configuration $G$. To do so, we construct a dedicated leader election algorithm for configuration $G$ using the canonical DRIP along with an appropriate decision function. 

	\begin{lemma}\label{YesCorrectness}
		If \classifier outputs ``Yes" for input $G$, then $G$ is a feasible configuration.
	\end{lemma}
\begin{proof}
	From the specification in Algorithm \ref{classifierpseudo}, \classifier outputs ``Yes" only if, for exactly one $j \in 1,\ldots,\lceil n/2 \rceil$, there exists an $m \in \{1,\ldots,numClasses_{G,j+1}\}$ such that exactly one node $v \in G_{aug}$ has $v_{\CLASS,j+1} = m$. Let $\hat{m}$ be the smallest such $m$ and let $\hat{v}$ be the node in class $\hat{m}$. By Lemma \ref{lem:diffClassesHistory}, for all $w \neq \hat{v}$, we have $\canonHist{\hat{v}}{r_j} \neq \canonHist{w}{r_j}$. As \classifier terminates after iteration $j$, it follows that $\mathcal{L}_{j+1}$ consists of the string ``terminate", by definition. So, by the definition of \canonDRIP, all nodes terminate in their local round $r_j+1$. Thus, $done_v = r_j+1$ for all $v$ in $G$. Define a decision function $f$ as follows: set $f(\canonHist{\hat{v}}{done_{\hat{v}}}) = 1$ and set $f(\canonHist{w}{done_w}) = 0$ for all $w \neq \hat{v}$. It follows that $(\canonDRIP,f)$ solves leader election with $\hat{v}$ as the unique leader.
	 
%	 To conclude the proof, we show that this function is well-defined, i.e., $\canonHist{v}{done_{v}} \neq \canonHist{w}{done_w}$ for all $w \neq v$. Recall that, according to the definition of the canonical DRIP, all nodes will terminate in the same local round: there is some $j'$ such that $m_{j'} = m_{j'-1}$ and all nodes will terminate in the first round of phase $P_{j'}$ in their local execution. It follows that $done_w = done_v$ for all $w \neq v$. However, as we know that $\canonHist{v}{r_j} \neq \canonHist{w}{r_j}$ for all $w \neq v$, it follows that $ \canonHist{v}{done_{v}}$ and $\canonHist{w}{done_w}$ differ in some prefix, as required.
\end{proof}

\subsubsection{Correctness of \classifier: ``No" Instances}\label{noInstances}
In this section, we prove the correctness of \classifier in cases where it outputs ``No" for input configuration $G$. It is sufficient to show that, if $G$ is feasible, then \classifier outputs ``Yes" on configuration $G$. At a high level, we prove that if a configuration $G$ is feasible, then at the start of some phase in the execution of the canonical DRIP by the nodes of $G$, there is a node with a unique history (which is sufficient to solve leader election). This fact is not immediately obvious: the fact that $G$ is feasible means that there is some dedicated leader election algorithm $\mathcal{A}$ for $G$, and the DRIP used in algorithm $\mathcal{A}$ might behave very differently than the canonical DRIP. Our proof has two main steps. First, we prove that if a configuration $G$ is feasible, then there exists a dedicated leader election algorithm for configuration $G$ whose DRIP is patient. This will allow us to restrict attention to patient DRIPs, which are much easier to reason about. Then, we prove that if there exists a dedicated leader election algorithm for configuration $G$ whose DRIP is patient, then there exists a dedicated leader election algorithm for configuration $G$ whose DRIP is the canonical DRIP.

We proceed by considering an arbitrary leader election algorithm for configuration $G$, and constructing a patient DRIP that can be used to solve leader election in configuration $G$. The idea is to construct a DRIP where each node starts its execution with a listening period of length $\sigma$ (which ensures that all nodes wake up spontaneously) and then have each node simulate the original algorithm. This simulation will begin in local round $\sigma+1$ if the node doesn't receive any messages in rounds $0,\ldots,\sigma$, or in an earlier round if it receives a message in one of the rounds $0,\ldots,\sigma$ (as it must simulate a forced wakeup).

\begin{lemma}\label{lem:existsPatient}
	For any configuration $G$, if $G$ is feasible then there exists a patient DRIP $D_{pat}$ and a decision function $f_{pat}$ such that $(D_{pat},f_{pat})$ is a dedicated leader election algorithm for $G$.
\end{lemma}
\begin{proof}
	Consider an arbitrary configuration $G$ and suppose that there exists a DRIP $D$ and a decision function $f$ such that $(D,f)$ is a dedicated leader election algorithm for configuration $G$. We define a new DRIP $D_{pat}$ executed locally at each node $w$: listen for each of the first $s_w = \min\{\sigma,rcv_w\}$ rounds after wakeup, where $rcv_w$ is the first local round in which a message is received, and then execute $D$ starting in round $s_w+1$. More formally, for each round $i > s_w$, the action performed by $w$ in its local round $i$ is given by $D_{pat}(\cH_{w,D_{pat}}[0 \ldots i-1]) = D(\cH_{w,D_{pat}}[s_w \ldots i-1])$.
	
	We also define a new decision function $f_{pat}$ that, when given the entire history of a node $w$ in the execution of $D_{pat}$, evaluates $f$ with the suffix of $w$'s history starting from round $s_w$. Formally, $f_{pat}(\cH_{w,D_{pat}}[0\ldots done_{w,D_{pat}}]) = f(\cH_{w,D_{pat}}[s_w \ldots done_{w,D_{pat}}])$. 
	
	First, we prove that $D_{pat}$ is a patient DRIP, i.e., that no node transmits in global rounds $0,\ldots,\sigma$.
	
	\begin{claim}\label{DpatPatient}
		In the execution of $D_{pat}$ by all nodes in $G$, no node transmits in any of the global rounds $0,\ldots,\sigma$. Equivalently, each node in $G$ wakes up spontaneously in the execution of $D_{pat}$.
	\end{claim}
		To prove the claim, we proceed by induction on the global round number $r$. In the base case, $r=0$, all nodes with wakeup tag greater than 0 are not awake, and those with wakeup tag equal to 0 start executing their local algorithm in global round 1. Thus, no node transmits in global round 0. As induction hypothesis, assume that no node transmits in global rounds $0,\ldots,k-1$ for some $k \in \{1,\ldots,\sigma\}$. We consider the behaviour of an arbitrary node $v$ in global round $k$. By the induction hypothesis, $v$ does not receive any messages in global rounds $0,\ldots,k-1$. We consider two cases:
	\begin{itemize}
		\item if $v$ has wakeup tag greater than or equal to $k$, then the fact that $v$ does not receive any messages in global rounds $0,\ldots,k-1$ implies that $v$ is not awake before global round $k$. Therefore, $v$ does not start executing its local algorithm until global round $k+1$ or later, so $v$ does not transmit in global round $k$.
		\item if $v$ has wakeup tag less than $k$, then, from the description of $D_{pat}$, node $v$ listens in global round $k$ as it has not received a message before global round $k$ and its local clock is less than $\sigma+1$. To justify this last fact, note that a local clock value is always bounded above by the global clock value, which in this case is $k \leq \sigma$.
	\end{itemize}
	In all cases, $v$ does not transmit in round $k$, so it follows that no node transmits in global rounds $0,\ldots,k$, which completes the inductive step. This completes the proof of Claim \ref{DpatPatient}.
	
	Next, we set out to prove that $(D_{pat},f_{pat})$ solves leader election when executed by configuration $G$. 
	We first prove a claim that essentially shows that the local experience of each node $w$ is the same when executing $D$ and $D_{pat}$, as long as we ignore the first $s_w$ rounds of $w$'s execution of $D_{pat}$. In other words, we show that, for each node $w$, the history of $w$'s execution of $D_{pat}$ starting at local round $s_w$ is identical to the history of $w$'s execution of $D$ starting at local round 0. This fact will later be used to show that, at each node, $f_{pat}$ outputs the same value as $f$, which we assumed evaluates to 1 for exactly one node in $G$. 
	 %whatever happens in global round $r$ in the execution of DRIP $D$ happens in global round $r + \sigma$ in the execution of DRIP $D_{pat}$. This fact will then be used to show that, other than an initial listening period, the history of each node in the executions of $D$ and $D_{pat}$ are the same. As $f_{pat}$ essentially chops off this listening period and then applies $f$, it will follow that $f_{pat}$ evaluates to 1 for exactly one node in $C$.
	
	\begin{claim}\label{DpatShift}
			Consider an arbitrary configuration $G$. Let $\psi$ be the execution of DRIP $D$ by the nodes of $G$, and let $\psi_{pat}$ be the execution of DRIP $D_{pat}$ by the nodes of $G$. For any $r \geq 0$:
			\begin{enumerate}[label=(\arabic*)]
				\item for all nodes $x$ in $G$, node $x$ transmits in global round $r$ of $\psi$ if and only if $x$ transmits in global round $r+\sigma$ of $\psi_{pat}$, and,\label{bulletShiftSigma}
				\item for all nodes $x$ in $G$, node $x$ wakes up in global round $r$ of $\psi$ if and only if, in $\psi_{pat}$, the value of $s_x$ is $r - t_x + \sigma$, and,\label{bulletWakeup}
				\item for all nodes $x$ in $G$, if $x$ is awake in global round $r$ of $\psi$, then $\cH_{x,D_{pat}}[s_x\ldots s_x+r^{(x)}] = \cH_{x,D}[0\ldots r^{(x)}]$ where $r^{(x)}$ is the local round at $x$ corresponding to global round $r$. \label{bulletHistory}
			\end{enumerate}
	\end{claim}
	We proceed by induction on the global round number $r$. For the base case, consider $r=0$.
	
	To prove \ref{bulletShiftSigma}, note that no node transmits in global round 0 of $\psi$. Further, by Claim \ref{DpatPatient}, no node transmits in global round $\sigma$ in $\psi_{pat}$. Therefore, both directions of the biconditional statement are vacuously true.
	
	To prove \ref{bulletWakeup}, consider an arbitrary node $w$ in $G$. First, suppose that $w$ wakes up in round 0 of $\psi$. It follows that $w$'s wakeup tag $t_w$ is 0: node $w$'s wakeup is not forced since no node transmits in global round 0 (as no node wakes up before round 0, and, by the definition of the model, no node transmits in its wakeup round). Since $t_w$ is 0, node $w$ wakes up spontaneously in global round 0 of $\psi_{pat}$ as well. By Claim \ref{DpatPatient}, no node transmits in global rounds $0,\ldots,\sigma$, so it follows that, in $\psi_{pat}$, the value of $s_w$ is $\min\{\sigma,rcv_w\} = \sigma$. Thus, $s_w = r - t_w + \sigma$, as required. Next, for the converse direction, suppose that $s_w = r - t_w + \sigma$ in $\psi_{pat}$. As $r=0$ and $D_{pat}$ is a patient DRIP, it follows that the global round corresponding to $w$'s local round $s_w$ in $\psi_{pat}$ is $s_w + t_w = \sigma$. By Claim \ref{DpatPatient}, no node transmits in global round $\sigma$ of $\psi_{pat}$, which implies that $s_w \neq rcv_w$, i.e., $s_w = \sigma$. Setting $\sigma = r - t_w + \sigma$, it follows that $t_w = 0$, which implies that $w$ wakes up in global round 0 of $\psi$, as required.

	To prove \ref{bulletHistory}, consider an arbitrary node $w$ in $G$. Suppose that node $w$ is awake in global round 0 of $\psi$. As no node wakes up before global round 0, it follows that $w$ wakes up in this round. Further, as no node transmits in the round it wakes up, node $w$ does not receive a message in global round 0 (i.e., its local round 0), so it follows that $\cH_{w,D}[0] = \emptyset$ and $t_w = 0$. As $t_w = 0$, node $w$ wakes up spontaneously in global round 0 of $\psi_{pat}$. By \ref{bulletWakeup}, in $\psi_{pat}$, the value of $s_w$ is $r-t_w+\sigma = \sigma$. But, as $w$ wakes up in global round 0 in $\psi_{pat}$, it follows that local round $\sigma$ at node $w$ is global round $\sigma$ in $\psi_{pat}$, and, by Claim \ref{DpatPatient}, no node transmits in global round $\sigma$ of $\psi_{pat}$. Therefore, $w$ does not receive a message in local round $\sigma$ in $\psi_{pat}$, and it follows that $\cH_{w,D_{pat}}[s_w] = \cH_{w,D_{pat}}[\sigma] = \emptyset = \cH_{w,D}[0]$, as required.
	
	As induction hypothesis, assume that the three statements of the claim hold for some $r \geq 0$. We proceed to prove the induction step for the three statements of the claim.
	
	To prove \ref{bulletShiftSigma}, consider an arbitrary node $w$ in $G$. We consider two cases:
    \begin{itemize}
    	\item Suppose that $w$ does not wake up in any of the global rounds $0,\ldots,r$ of $\psi$. Note that this means that $w$ does not transmit in global round $r+1$ of $\psi$ (as $r+1$ is its earliest possible wakeup round and no node transmits in its wakeup round) so one direction of the biconditional statement holds vacuously. For the other direction, note that if $w$ does not wake up in any of the global rounds $0,\ldots,r$, then by statement \ref{bulletWakeup} of the induction hypothesis, it follows that $s_{w} > r - t_{w} + \sigma$, i.e., $s_{w} \geq (r+1) - t_{w} + \sigma$. By the definition of $D_{pat}$, node $w$ listens in the $s_{w}$ rounds after wakeup. In particular, this means that $w$ does not transmit in its local round $(r+1)-t_{w} + \sigma$. But, by Claim \ref{DpatPatient}, node $w$ wakes up spontaneously, which occurs in global round $t_{w}$ by the definition of the model, so its local round $(r+1)-t_{w} + \sigma$ corresponds to global round $r+1+\sigma$. Thus, $w$ does not transmit in global round $r+1+\sigma$ of $\psi_{pat}$, as required.
      	\item Suppose that $w$ wakes up in some global round $r' \in \{0,\ldots,r\}$ of $\psi$. By statement \ref{bulletWakeup} of the induction hypothesis, it follows that $s_{w} = r' - t_{w} + \sigma$. Also, by statement \ref{bulletHistory} of the induction hypothesis, it follows that $\cH_{w,D_{pat}}[s_{w}\ldots s_{w}+r^{(w)}] = \cH_{w,D}[0\ldots r^{(w)}]$, where $r^{(w)}$ is the local round at $w$ corresponding to global round $r$ of $\psi$. By the definition of $D_{pat}$ and the two preceding facts, we get that
      	\begin{align*}
      	D_{pat}(\cH_{w,D_{pat}}[0\ldots r-t_{w}+\sigma]) & = D(\cH_{w,D_{pat}}[s_{w} \ldots r-t_{w}+\sigma])\\
      	& = D(\cH_{w,D_{pat}}[s_{w}\ldots r'-t_{w}+\sigma + (r-r')])\\
      	& = D(\cH_{w,D_{pat}}[s_{w}\ldots s_{w} + (r-r')])\\
      	& = D(\cH_{w,D}[0\ldots (r-r')])
      	\end{align*}
      	In other words, the above equality means that the action performed by $w$ in local round $r-t_{w} + \sigma + 1$ in its execution of DRIP $D_{pat}$ is equal to the action performed by $w$ in local round $r-r'+1$ in its execution of DRIP $D$. As $D_{pat}$ is a patient DRIP, we know that $w$ wakes up spontaneously, i.e., in global round $t_{w}$ of $\psi_{pat}$, and, by the definition of $r'$, we know that $w$ wakes up in global round $r'$ of $\psi$. So, translating from local rounds to global rounds, the above statement is equivalent to the following: the action performed by $w$ in global round $r + \sigma + 1$ in its execution of DRIP $D_{pat}$ is equal to the action performed by $w$ in global round $r+1$ in its execution of DRIP $D$. This implies the desired result.
    \end{itemize} 
	As the statement holds in both cases, this concludes the proof of statement \ref{bulletShiftSigma}.
	
	To prove statement \ref{bulletWakeup}, consider an arbitrary node $w$ in $G$. There are several cases to consider:
	\begin{itemize}
		\item Suppose that node $w$ does not wake up in global round $r+1$ of $\psi$. There are two possibilities:
		\begin{itemize}
			\item {\bf $w$ wakes up in some global round $r' < r+1$ of $\psi$.} By statement \ref{bulletWakeup} of the induction hypothesis, the value of $s_w$ is $r' - t_w + \sigma$, which does not equal $(r+1) - t_w + \sigma$, as desired.
			\item {\bf $w$ wakes up after global round $r+1$ of $\psi$.} It follows that $r+1 < t_w$ and that node $w$ does not receive a message in global round $r+1$ of $\psi$. To obtain a contradiction, assume that $s_w = (r+1)-t_w+\sigma$ in $\psi_{pat}$. Recall that $s_w$ is defined as $\min\{\sigma,rcv_w\}$, where $rcv_w$ is the first round in which $w$ receives a message. The fact that $r+1 < t_w$ rules out the possibility that $s_w = \sigma$. The remaining possibility is that $s_w = rcv_w$, i.e., $w$ receives a message in local round $s_w = (r+1)-t_w+\sigma$ of its execution of $D_{pat}$. We obtain a contradiction by showing that this implies that $w$ receives a message in global round $r+1$ of $\psi$. Indeed, as $D_{pat}$ is a patient DRIP, $w$ wakes up spontaneously in round $t_w$, so $w$'s local round $s_w$ corresponds to global round $s_w + t_w = (r+1) + \sigma$ in $\psi_{pat}$. As $w$ receives a message in local round $s_w$, it follows that $w$ listens and has exactly one neighbour in $G$ that transmits in global round $(r+1)+\sigma$ of $\psi_{pat}$. By statement \ref{bulletShiftSigma} applied to $w$ and all of its neighbours, it follows that node $w$ listens and has exactly one neighbour in $G$ that transmits in global round $r+1$ of $\psi$, so $w$ receives a message in global round $r+1$ of $\psi$. As this contradiction was reached under the assumption that $s_w = (r+1)-t_w+\sigma$ in $\psi_{pat}$, it follows that $s_w \neq (r+1)-t_w+\sigma$, as desired.
		\end{itemize}
		\item Suppose that node $w$ wakes up in global round $r+1$ of $\psi$. There are two possibilities:
		\begin{itemize}
			\item {\bf $w$ wakes up spontaneously in global round $r+1$ of $\psi$.} Then, by definition, $t_w = r+1$, and $w$ does not receive a message in any of the global rounds $0,\ldots,r+1$ of $\psi$. By the induction hypothesis and the proof above, statement \ref{bulletShiftSigma} holds for $w$ and all neighbours of $w$ for each of the global rounds $0,\ldots,r+1$ of $\psi$, so it follows that $w$ does not receive a message in any of the global rounds $\sigma,\ldots,r+1+\sigma$ of $\psi_{pat}$. By Claim \ref{DpatPatient}, no nodes transmit in global rounds $0,\ldots,\sigma$ of $\psi_{pat}$, so, along with the preceding fact, we have that $w$ does not receive a message in any of the global rounds up to and including $r+1+\sigma$. As $D_{pat}$ is a patient DRIP, $w$ wakes up in global round $t_w$ of $\psi_{pat}$, so global round $r+1+\sigma$ corresponds to $w$'s local round $r+1+\sigma-t_w = \sigma$. Thus, $w$ does not receive a message in any round up to and including its local round $\sigma$, which implies that $s_w = \min\{\sigma,rcv_w\}=\sigma = (r+1)-t_w+\sigma$, as desired.
			\item {\bf $w$ wakes up due to receiving a message from a neighbour in global round $r+1$ of $\psi$.} 
			As $w$ did not spontaneously wake up before round $r+1$, it follows that $t_w \geq r+1$. Also, as $w$ wakes up in global round $r+1$, we know that $w$ does not receive a message in any of the global rounds $0,\ldots,r$ of $\psi$, and receives a message in global round $r+1$ of $\psi$. By the induction hypothesis and the proof above, statement \ref{bulletShiftSigma} holds for $w$ and all neighbours of $w$ for each of the global rounds $0,\ldots,r+1$ of $\psi$, so it follows that $w$ does not receive a message in any of the global rounds $\sigma,\ldots,r+\sigma$ of $\psi_{pat}$, and receives a message in global round $r+1+\sigma$ of $\psi_{pat}$. By Claim \ref{DpatPatient}, no nodes transmit in global rounds $0,\ldots,\sigma$ of $\psi_{pat}$, so, along with the preceding fact, we have that $w$ does not receive a message in any of the global rounds up to and including $r+\sigma$, and receives a message in global round $r+1+\sigma$ of $\psi_{pat}$.
		 As $D_{pat}$ is a patient DRIP, $w$ wakes up spontaneously in global round $t_w$ of $\psi_{pat}$, so global round $r+1+\sigma$ corresponds to $w$'s local round $r+1+\sigma-t_w$. In particular, we have shown that $w$'s local round $r+1+\sigma-t_w$ is the first round in which $w$ receives a message in $\psi_{pat}$, i.e., $rcv_w = r+1+\sigma-t_w$. As $t_w \geq r+1$, it follows that $rcv_w = r+1+\sigma-t_w \leq \sigma$, so $\min\{\sigma,rcv_w\} = rcv_w$. Thus, $s_w = rcv_w = (r+1)-t_w+\sigma$, as desired.
		\end{itemize}
	\end{itemize}

This concludes the proof of statement \ref{bulletWakeup}. 

To prove statement \ref{bulletHistory}, consider an arbitrary node $w$ in $G$. Suppose that $w$ is awake in global round $r+1$ of $\psi$. Let $r' \leq r+1$ be the global round in $\psi$ in which $w$ wakes up. It follows that global round $r+1$ corresponds to local round $r+1-r'$ at node $w$. So, we must prove that $\cH_{w,D_{pat}}[s_w\ldots s_w+(r+1-r')] = \cH_{w,D}[0\ldots r+1-r']$. In fact, it suffices to show that $\cH_{w,D_{pat}}[s_w+(r+1-r')] = \cH_{w,D}[r+1-r']$, as we already know that the remainder of the entries are equal, either due to statement \ref{bulletHistory} of the induction hypothesis, or in the case that $r'=r+1$. In other words, we must show that the history of $w$ in its local round $r+1-r'$ of $\psi$ is the same as the history of $w$ in its local round $s_w+(r+1-r')$ of $\psi_{pat}$. To do so, we first state this equivalently with respect to global rounds. As node $w$ wakes up in round $r'$ of $\psi$, it follows that local round $r+1-r'$ at node $w$ corresponds to global round $r+1$ of $\psi$. As $D_{pat}$ is a patient DRIP, node $w$ wakes up spontaneously in round $t_w$ of $\psi_{pat}$, so local round $s_w+(r+1-r')$ at node $w$ corresponds to global round $s_w+(r+1-r')+t_w$ of $\psi_{pat}$. However, by statement \ref{bulletWakeup}, the value of $s_w$ is $r' - t_w + \sigma$, so $s_w+(r+1-r')+t_w = r+1+\sigma$. So, equivalently, we must prove that the history of $w$ in global round $r+1$ of $\psi$ is the same as the history of $w$ in global round $r+1+\sigma$ of $\psi_{pat}$. However, this statement is seen to be true by applying statement \ref{bulletShiftSigma} to $w$ and all of $w$'s neighbours. This concludes the proof of statement \ref{bulletHistory}.
% in $\ If $w$ wakes up in global round $r+1$, then the corresponding local round $r_w+1$ is 0, and so it only remains to show that $\cH_{w,D_{pat}}[s_w] = \cH_{w,D}[0]$. 

%By statement 

%	For the induction step, consider an arbitrary node $w$ in $C$ that is awake in global round $r+1$ in the execution of $D$ by all nodes. Consider an arbitrary neighbour $\hat{w}$ of $w$ in $G$ that is awake in global round $r+1$ in the execution of $D$ by all nodes. There are two cases to consider:
%	\begin{itemize}
%		\item Suppose that $\hat{w}$ was awake in global round $r$. By the induction hypothesis, $\cH_{\hat{w},D}[0\ldots r_{\hat{w}}] = \cH_{\hat{w},D_{pat}}[s_{\hat{w}}\ldots s_{\hat{w}}+r_{\hat{w}}]$.
%		\item Suppose that $\hat{w}$ wakes up in global round $r+1$.
%	\end{itemize}

	%Further, by the definition of $D_{pat}$, the value of $s_w$ is $\min\{\sigma,\textrm{rcv}_w\}$
	%corresponding to global round $r$ is  and consider $w$'s history from local round $s_w$ up to local round $r-s_w$, i.e., $\cH_{w,D_{pat}}[s_w\ldots r-s_w]$.    
	%{\color{red}[fill in details]}
	As the three statements of the Claim have been proven, the induction step is complete, which concludes the proof of Claim \ref{DpatShift}.

	Next, from statement \ref{bulletHistory} of Claim \ref{DpatShift} with $x = w$ and $r^{(w)} = done_{w,D}$, we get that
\begin{equation}\label{sameHistD}
\cH_{w,D}[0\ldots done_{w,D}] = \cH_{w,D_{pat}}[s_w\ldots s_w + done_{w,D}]
\end{equation}
 for any node $w$ in $G$. Note that $s_w + done_{w,D} = done_{w,D_{pat}}$: in the execution of $D$ by all nodes, $done_{w,D}$ is the first round in which the DRIP $D$ outputs $terminate$, and since $D_{pat}$ is executing $D$ starting in round $s_w$ and the history is the same up to round $s_w + done_{w,D}$, the DRIP $D_{pat}$ will first output $terminate$ in round $s_w + done_{w,D}$. Therefore, it follows from equation (\ref{sameHistD}) that
 \begin{equation}\label{sameHistDpat}
\cH_{w,D}[0\ldots done_{w,D}] = \cH_{w,D_{pat}}[s_w\ldots done_{w,D_{pat}}]
\end{equation} 
for any node $w$ in $G$.
	
	 % into $\cH_{w,D}[0\ldots i] = \cH_{w,D_{pat}}[s_w\ldots s_w+i]$, and then rewriting $s_w+done_{w,D}$ as $done_{w,D_{pat}}$, we have shown that . 
	  
	  Using equation (\ref{sameHistDpat}) along with the definition of $f_{pat}$, it follows that 
	  \begin{align*}
	  f_{pat}(\cH_{w,D_{pat}}[0\ldots done_{w,D_{pat}}]) 
	  & = f(\cH_{w,D_{pat}}[s_w\ldots done_{w,D_{pat}}])\\
	  & = f(\cH_{w,D}[0\ldots done_{w,D}])
	  \end{align*}
	  for an arbitrary node $w$. As $f(\cH_{w,D}[0\ldots done_{w,D}]) = 1$ for exactly one node $w$ in $G$, it follows that $f_{pat}(\cH_{w,D_{pat}}[0\ldots done_{w,D_{pat}}]) = 1$ for exactly one node $w$ in $G$, which completes the proof that $(D_{pat},f_{pat})$ is a leader election algorithm for configuration $G$.
	\end{proof}

%For example, in an arbitrary DRIP, a node $v$'s transmission might result in a forced wakeup of a node $w$, whereas the canonical DRIP is designed in such a way that this will not happen (it is a patient DRIP). 

Lemma \ref{lem:existsPatient} shows that if a configuration $G$ is feasible, then there exists a patient DRIP that can be used to solve leader election in $G$. The main goal of the remainder of the proof is to consider any patient DRIP used in a leader election algorithm for a configuration $G$, and show that the canonical DRIP can do no worse at breaking symmetry between nodes in $G$. In particular, if two nodes have different histories at some time during the execution of an arbitrary patient DRIP, then they will have different histories at the start of some phase during the execution of the canonical DRIP.

The following result will be helpful in proving that two nodes $v$ and $w$ have different histories in the execution of the canonical DRIP. In particular, if a neighbour $\hat{v}$ of $v$ transmits in $v$'s local round $r$ and a neighbour $\hat{w}$ of $w$ transmits in $w$'s local round $r' \neq r$, then this will continue to happen in all future phases. Essentially, once a neighbour of $v$ and a neighbour of $w$ have been ``separated" with respect to when they transmit, then they ``stay separated". 
%This will allow us to show that, eventually, there will be a phase where the histories of $v$ and $w$ are different. 
%It is then sufficient to  Therefore, it will be sufficient to separately show that every such pair of neighbours of $v$ and $w$ are distinguished in this way in some phase  At a high level, the next result shows that two nodes that can be distinguished in some phase $P_{j'}$ can also be distinguished in all subsequent phases. In particular, if a node $v$ has a neighbour $\hat{v}$ and a node $w$ has a neighbour $\hat{w}$ such that 

\begin{lemma}\label{lem:staySeparate}
	Consider any nodes $v,w$ in $G$, let $\hat{v}$ be any neighbour of $v$, and let $\hat{w}$ be any neighbour of $w$. Suppose that there exists a $j' \geq 1$ such that, in the execution of the canonical DRIP, the local round in phase $P_{j'}$ at node $w$ in which $\hat{w}$ transmits is different than the local round in phase $P_{j'}$ at node $v$ in which $\hat{v}$ transmits. Then, for all $j \geq j'$, in the execution of the canonical DRIP, the local round in phase $P_{j}$ at node $w$ in which $\hat{w}$ transmits is different than the local round in phase $P_{j}$ at node $v$ in which $\hat{v}$ transmits.
\end{lemma}
\begin{proof}
	Consider any $j' \geq 1$ such that, in the execution of $\canonDRIP$, the local round in phase $P_{j'}$ at node $w$ in which $\hat{w}$ transmits is different than the local round in phase $P_{j'}$ at node $v$ in which $\hat{v}$ transmits. Define $h_v,k_v$ such that $\hat{v}$ transmits in the $h_v$'th round of the $k_v$'th transmission block in phase $P_{j'}$ of node $v$'s execution of the canonical DRIP, and define $h_w,k_w$ such that $\hat{w}$ transmits in the $h_w$'th round of the $k_w$'th transmission block in phase $P_{j'}$ of node $w$'s execution of the canonical DRIP. By the choice of $j'$, it must be the case that $h_v \neq h_w$ or $k_v \neq k_w$. 
	
	First, suppose that $k_v \neq k_w$. Therefore, there are two different equivalence classes $k_v,k_w$ with representatives $reps_{j'}[k_v]$ and $reps_{j'}[k_w]$, respectively. By Lemma \ref{lem:diffClassesHistory}, it follows that $\canonHist{reps_{j'}[k_v]}{r_{j'-1}} \neq \canonHist{reps_{j'}[k_w]}{r_{j'-1}}$. By statement \ref{kClassBlocks} of Lemma \ref{relateClassifierCanonical} and the definitions of $k_v$ and $k_w$, we know that $\canonHist{\hat{v}}{r_{j'-1}} = \canonHist{reps_{j'}[k_v]}{r_{j'-1}}$ and $\canonHist{\hat{w}}{r_{j'-1}} = \canonHist{reps_{j'}[k_w]}{r_{j'-1}}$, so it follows that $\canonHist{\hat{v}}{r_{j'-1}} \neq \canonHist{\hat{w}}{r_{j'-1}}$. For any $j \geq j'$, the fact that $\canonHist{\hat{v}}{r_{j'-1}} \neq \canonHist{\hat{w}}{r_{j'-1}}$ implies that $\canonHist{\hat{v}}{r_{j-1}} \neq \canonHist{\hat{w}}{r_{j-1}}$, since the prefixes of length $r_{j'-1}$ are different. So there cannot be a single value of $k$ such that $\canonHist{\hat{v}}{r_{j-1}} = \canonHist{reps_{j}[k]}{r_{j-1}}$ and $\canonHist{\hat{w}}{r_{j-1}} = \canonHist{reps_{j}[k]}{r_{j-1}}$, which means that $\hat{v}$ and $\hat{w}$ transmit in different transmission blocks of phase $P_j$ of the executions of nodes $v$ and $w$, respectively, and thus their transmissions cannot occur in the same local rounds at $v$ and $w$, respectively.
	
	Next, suppose that $h_v \neq h_w$. By Lemma \ref{lem:transmitConversions}, $h_v = t_{\hat{v}} - t_v + \sigma+1$ and $h_w = t_{\hat{w}} - t_w + \sigma+1$. In particular, note that the values of $h_v$ and $h_w$  depend only on the relative wakeup times of $v,\hat{v}$ and $w,\hat{w}$. In other words, every time $\hat{v}$ transmits, it will happen in local round $h_v$ at $v$ within some transmission block of some phase, and every time $\hat{w}$ transmits, it will happen in local round $h_w$ at $w$ within some transmission block of some phase. So, even if the transmissions by $\hat{v}$ and $\hat{w}$ in phase $P_{j}$ occur in the same transmission block locally at nodes $v$ and $w$, respectively, the fact that $h_v \neq h_w$ means that they do not have the same offset from the start of the transmission block, so they do not occur in the same local rounds at $v$ and $w$, respectively.
\end{proof}

We are now ready to prove the central fact that will be used to show that the canonical DRIP for a feasible configuration $G$ can be used to solve leader election in $G$: if two nodes $v$ and $w$ have different histories up to some round in the execution of some patient DRIP, then $v$ and $w$ will have different histories at the start of some phase $P_j$ during the execution of the canonical DRIP. 

\begin{lemma}\label{lem:alsoDistinguish}
	For any configuration $G$, any patient DRIP $D$, any $i \geq 0$, and any two nodes $v,w$ in $G$, if $\DHist{v}{i} \neq \DHist{w}{i}$, then in the execution of the canonical DRIP $\canonDRIP$, there exists a $j \geq 1$ such that $\canonHist{v}{r_{j-1}} \neq \canonHist{w}{r_{j-1}}$.
\end{lemma}
\begin{proof}
	Consider an arbitrary configuration $G$ and any patient DRIP $D$. 
	The proof is by induction on the global round number $r$. In particular, for each value $r \geq 0$, we prove that for any two nodes $v,w$ in $G$ that are both awake in global round $r$, if $i = r - \max\{t_v,t_w\} \geq 0$, then $\DHist{v}{i} \neq \DHist{w}{i}$ implies that there exists $j \geq 1$ such that $\canonHist{v}{r_{j-1}} \neq \canonHist{w}{r_{j-1}}$ in the execution of the canonical DRIP $\canonDRIP$. Doing so is sufficient to prove the result since for an arbitrary pair of nodes $v,w$ in $G$ and any $i \geq 0$, there exists a value for $r$ such that $v$ and $w$ are both awake in global round $r$ and such that $i = r-\max\{t_v,t_w\}$. 
	
%	The following fact will be used in several places to take care of node pairs $v,w$ where at least one of $v$ or $w$ wakes up in global round $r$. In particular, if at least one of $v$ or $w$ wakes up in global round $r$, then $i = r - \max\{t_v,t_w\} = 0$, and $\mathcal{H}_{v,D}[0] = (\emptyset) = \mathcal{H}_{w,D}[0]$ since $D$ is a patient DRIP, and desired statement is vacuously true.
%	
%	\begin{fact}\label{TrueInWakeupRound}
%		Consider any global round $r$, any two nodes $v,w$ in $C$, and let $i = r - \max\{t_v,t_w\}$. If at least one of $v$ or $w$ wakes up in global round $r$,  then $\DHist{v}{i} \neq \DHist{w}{i}$ implies that there exists $j \geq 0$ such that $v \not\equiv_{j} w$ in the execution of the canonical DRIP.
%	\end{fact}

	For the base case of the induction argument, consider $r=0$ and any $v,w$ in $G$ that are both awake in global round $r$. Since no node wakes up before global round 0, it follows that both $v$ and $w$ wake up in global round $r$, so $i = r - \max\{t_v,t_w\} = 0$. Since $D$ is a patient DRIP, both $v$ and $w$ wake up spontaneously, so $\mathcal{H}_{v,D}[0] = (\emptyset) = \mathcal{H}_{w,D}[0]$, and the desired statement is vacuously true.
	
	As induction hypothesis, assume that for all $k \in \{0,\ldots,r\}$, for any $v,w$ in $G$ that are both awake in global round $k$, if $i = k - \max\{t_v,t_w\} \geq 0$, then $\DHist{v}{i} \neq \DHist{w}{i}$ implies that there exists $j \geq 1$ such that $\canonHist{v}{r_{j-1}} \neq \canonHist{w}{r_{j-1}}$ in the execution of the canonical DRIP $\canonDRIP$.
	
	For the induction step, consider any $v,w$ in $G$ that are both awake in global round $r+1$. Suppose that $i = r + 1 - \max\{t_v,t_w\} \geq 0$, and suppose that $\DHist{v}{i} \neq \DHist{w}{i}$. 
	
	First, consider the case where the histories of $v$ and $w$ differ before round $i$ in the execution of $D$, i.e., assume there exists $i' \in \{0,\ldots,i-1\}$ such that $\DHist{v}{i'} \neq \DHist{w}{i'}$. Then, the fact that $i' < i$ means that $i' = k - \max\{t_v,t_w\}$ for some $k \leq r$, so the induction hypothesis implies that there exists $j' \geq 1$ such that $\canonHist{v}{r_{j'-1}} \neq \canonHist{w}{r_{j'-1}}$ in the execution of the canonical DRIP $\canonDRIP$, as desired. 
	%If $j' \geq i$, then the claim holds for $j = j'$. Otherwise, if $j' < i$, then by definition, $v \not\equiv_{{j'}} w$ means that $\canonHist{v}{r_{j'}} \neq \canonHist{w}{r_{j'}}$, and since $r_i > r_{j'}$, it follows that $\canonHist{v}{r_{i}} \neq \canonHist{w}{r_{i}}$, i.e., $v \not\equiv_{{i}} w$, and the claim holds for $j = i$.
	
	So the remainder of the proof assumes that $i$ is the first local round where the histories of $v$ and $w$ differ in the execution of $D$, i.e., $\DHist{v}{i-1} = \DHist{w}{i-1}$ and $\mathcal{H}_{v,D}[i] \neq \mathcal{H}_{w,D}[i]$. Since $\DHist{v}{i-1} = \DHist{w}{i-1}$, we know that $v$ and $w$ perform the same action in local round $i$ of the DRIP $D$. Further, we know that $v$ and $w$ must both listen in their local round $i$ in the execution of $D$, since otherwise we would have $\mathcal{H}_{v,D}[i] = (\emptyset) = \mathcal{H}_{w,D}[i]$, which contradicts our assumption that $\mathcal{H}_{v,D}[i] \neq \mathcal{H}_{w,D}[i]$. 
	
	Next, we prove a useful claim that, at a high level, shows that if a neighbour $\hat{w}$ of $w$ behaves differently than a neighbour $\hat{v}$ of $v$ in the same local round $i$ in the execution of $D$, then there is a phase in the canonical DRIP where $\hat{w}$ transmits in a different local round than $\hat{v}$ does. This will help us conclude that if $v$ and $w$ have different histories in local round $i$ of $D$, then this difference will be noticed in some phase of the canonical DRIP as well.
	
	\begin{claim}\label{diffBehaviour}
		Let $\hat{v}$ be an arbitrary neighbour of $v$ in $G$ and let $\hat{w}$ be an arbitrary neighbour of $w$ in $G$. Suppose that $\hat{v}$ transmits a message $M$ in $v$'s local round $i$ in the execution of $D$. Suppose that, in $w$'s local round $i$ in the execution of $D$, node $\hat{w}$ does not transmit, or transmits a message $M' \neq M$. Then there exists a $j' \geq 1$ such that, in the execution of $\canonDRIP$, the local round in phase $P_{j'}$ at node $w$ in which $\hat{w}$ transmits is different than the local round in phase $P_{j'}$ at node $v$ in which $\hat{v}$ transmits.
	\end{claim}
	\noindent Proof of the Claim: 
	%By Lemma \ref{lem:localconversion}, the transmission by $\hat{v}$ during $v$'s local round $i$ in the execution of $D$ occurs during $\hat{v}$'s local round $i-(t_{\hat{v}}-t_{v})$ in the execution of $D$. Consider phase $P_{i-(t_{\hat{v}}-t_{v})}$ in $v$'s execution of the canonical DRIP. By Lemma \ref{lem:eachTransmits}, there exists an $h \in \{1,\ldots,2\sigma+1\}$ and $k \in \{1,\ldots,m_{i-(t_{\hat{v}}-t_{v})-1}\}$ such that $\hat{v}$ transmits in a round corresponding to the $h$'th round of the $k$'th transmission block of phase $P_{i-(t_{\hat{v}}-t_{v})}$ in $v$'s execution of the canonical DRIP. This is local round number $r_{i-(t_{\hat{v}}-t_{v})-1} + (k-1)(2\sigma+1) + h$ at $v$. Similarly, by Lemma \ref{lem:eachTransmits}, $\hat{w}$ transmits in some round during phase $P_{i-(t_{\hat{v}}-t_{v})}$ in $w$'s execution of the canonical DRIP. We set out to show that this transmission does not occur in local round number $r_{i-(t_{\hat{v}}-t_{v})-1} + (k-1)(2\sigma+1) + h$ at $w$. 
	%(In particular, as this applies to an arbitrary neighbour $\hat{w}$ of $w$, we will use it to show that $v$'s history in phase $P_{i-(t_{\hat{v}}-t_{v})}$ of the canonical DRIP is different than $w$'s history in phase $P_{i-(t_{\hat{v}}-t_{v})}$ of the canonical DRIP.) 
	We proceed in cases depending on the relationship between $t_{\hat{v}}-t_v$ and $t_{\hat{w}}-t_w$.
	\begin{itemize}
		\item Suppose that $t_{\hat{v}}-t_v \neq t_{\hat{w}}-t_w$. Consider phase $P_1$. By Lemma \ref{lem:transmitConversions}, node $\hat{v}$'s transmission occurs in $v$'s local round $r_{0} + (k-1)(2\sigma+1) + h = r_{0} + (k-1)(2\sigma+1) + (t_{\hat{v}} - t_v + \sigma + 1)$. To obtain a contradiction, assume that $\hat{w}$ transmits in a round corresponding to the same local round at $w$, i.e., the $h$'th round of the $k$'th transmission block of phase $P_{0}$ in $w$'s execution of the canonical DRIP. By Lemma \ref{lem:transmitConversions}, the transmission by $\hat{w}$ occurs in $w$'s local round $r_{0} + (k-1)(2\sigma+1) + (t_{\hat{w}} - t_w + \sigma + 1)$. But if the local round at $v$ and the local round at $w$ are equal, then $r_{0} + (k-1)(2\sigma+1) + (t_{\hat{v}} - t_v + \sigma + 1) = r_{0} + (k-1)(2\sigma+1) + (t_{\hat{w}} - t_w + \sigma + 1)$, which would imply that $t_{\hat{v}}-t_v = t_{\hat{w}}-t_w$, a contradiction. So our assumption was incorrect, i.e., for $j' = 1$, it must be the case that $w$'s local round in phase $P_{j'}$ in which $\hat{w}$ transmits is different than $v$'s local round in phase $P_{j'}$ in which $\hat{v}$ transmits.
		
		\item Suppose that $t_{\hat{v}}-t_v = t_{\hat{w}}-t_w$. At a high level, we proceed by applying the induction hypothesis to the transmissions by $\hat{v}$ and $\hat{w}$, which will imply that they are placed in different equivalence classes in some phase of the canonical DRIP, and thus transmit during different transmission blocks, which will correspond to different local rounds at $v$ and $w$.
		
		We show that the conditions of the induction hypothesis hold for nodes $\hat{v}$ and $\hat{w}$ in their local rounds $i-(t_{\hat{v}}-t_{v})-1$ in the execution of $D$.
		
		First, we show that both $\hat{v}$ and $\hat{w}$ are awake in global round $r$ in the execution of $D$. First, suppose that $t_{v} \geq t_{w}$. Since $t_{\hat{v}} - t_v = t_{\hat{w}} - t_w$, it follows that $t_{\hat{v}} \geq t_{\hat{w}}$. Since $i = r+1-\max\{t_v,t_w\} = r+1-t_v$, we know that $v$'s local round $i$ corresponds to global round $r+1$. In particular, this means that $\hat{v}$ is awake and transmits in global round $r+1$ in the execution of $D$, and thus is awake in round $r$ (since, in our model, no node transmits in the same round as it wakes up). Since $t_{\hat{v}} \geq t_{\hat{w}}$, it follows that $\hat{w}$ is also awake in global round $r$ in the execution of $D$.
		Next, suppose that $t_{w} > t_{v}$. Since $t_{\hat{v}} - t_v = t_{\hat{w}} - t_w$, it follows that $t_{\hat{w}} > t_{\hat{v}}$. Since $i = r+1-\max\{t_v,t_w\} = r+1-t_w$, we know that $w$'s local round $i$ corresponds to global round $r+1$. As node $v$ spontaneously wakes up in global round $t_v$, it follows that $v$'s local round $i$ occurs in global round $t_v + i = t_v + (r+1-t_w) < r+1$, where the last inequality is due to the fact that $t_w > t_v$. In particular, as $\hat{v}$ transmits during $v$'s local round $i$, this means that node $\hat{v}$ transmits in or before global round $r$ in the execution of $D$. As $D$ is a patient DRIP, no node transmits in global rounds $0,\ldots,\sigma$, which implies that $r > \sigma$. Moreover, every node wakes up spontaneously in the round equal to their wakeup tag, so $t_{\hat{w}} \leq \sigma < r$, which implies that $\hat{w}$ is awake in global round $r$.  Since $t_{\hat{w}} > t_{\hat{v}}$, it follows that $\hat{v}$ is also awake in global round $r$. This concludes the proof that both $\hat{v}$ and $\hat{w}$ are awake in global round $r$ in the execution of $D$.
		%their local round number that corresponds to local round $i-1$ at $v$.  
		
		Next, we show that $i-(t_v-t_{\hat{v}})-1 = r - \max\{t_{\hat{v}},t_{\hat{w}}\}$. If $t_{v} \geq t_{w}$, then since $i = r+1-\max\{t_{v},t_{w}\}$, we get that $i-(t_{\hat{v}}-t_v)-1 = (r+1-\max\{t_{v},t_{w}\}) -(t_{\hat{v}}-t_v)-1 = r - t_{\hat{v}}$. Since $t_{\hat{v}} - t_v = t_{\hat{w}} - t_w$, it follows that $t_{\hat{v}} \geq t_{\hat{w}}$, so $r-t_{\hat{v}}$ is equal to $r - \max\{t_{\hat{v}},t_{\hat{w}}\}$.  If $t_{w} \geq t_{v}$, then note that $i - (t_{\hat{v}}-t_{v}) -1 = i - (t_{\hat{w}}-t_{w})-1$ from the assumption that $t_{\hat{v}} - t_v = t_{\hat{w}} - t_w$. Then, since $i = r+1-\max\{t_{v},t_{w}\}$, we get that $i - (t_{\hat{w}}-t_{w})-1 = (r+1-\max\{t_{v},t_{w}\}) - (t_{\hat{w}}-t_{w})-1 = r - t_{\hat{w}}$. Since $t_{\hat{v}} - t_v = t_{\hat{w}} - t_w$, it follows that $t_{\hat{w}} \geq t_{\hat{v}}$, so $r-t_{\hat{w}}$ is equal to $r - \max\{t_{\hat{v}},t_{\hat{w}}\}$. This concludes the proof that $i-(t_v-t_{\hat{v}})-1 = r - \max\{t_{\hat{v}},t_{\hat{w}}\}$.
		
		Finally, we prove that $\DHist{\hat{v}}{i-(t_{\hat{v}}-t_v)-1} \neq \DHist{\hat{w}}{i-(t_{\hat{w}}-t_w)-1}$. To obtain a contradiction, assume otherwise. Then nodes $\hat{v}$ and $\hat{w}$ would perform the same action in their local rounds $i-(t_{\hat{v}}-t_v)$ and $i-(t_{\hat{w}}-t_w)$, respectively, in their execution of $D$. By Proposition \ref{prop:localconversion}, these rounds correspond to local round $i$ at $v$ and local round $i$ at $w$, respectively, in their execution of $D$. By assumption, $\hat{v}$ sends message $M$ in local round $i$ at $v$, so $\hat{w}$ would also send message $M$ in local round $i$ at $w$. This contradicts the assumption that either $\hat{w}$ does not transmit, or transmits a message $M' \neq M$. This concludes the proof that $\DHist{\hat{v}}{i-(t_{\hat{v}}-t_v)-1} \neq \DHist{\hat{w}}{i-(t_{\hat{w}}-t_w)-1}$.
		
		Altogether, we have shown that $\hat{v}$ and $\hat{w}$ are awake in global round $k=r$, and for $i-(t_{\hat{v}}-t_v)-1 = i-(t_{\hat{w}}-t_w)-1 = k - \max\{t_{\hat{v}},t_{\hat{w}}\}$, we have $\DHist{\hat{v}}{i-(t_{\hat{v}}-t_v)-1} \neq \DHist{\hat{w}}{i-(t_{\hat{w}}-t_w)-1}$. So, by the induction hypothesis, there exists $j' \geq 1$ such that $\canonHist{\hat{v}}{r_{j'-1}} \neq \canonHist{\hat{w}}{r_{j'-1}}$ in the execution of the canonical DRIP $\canonDRIP$. So, by Lemma \ref{lem:diffClassesHistory}, it follows that $\hat{v}_{\CLASS,j'} \neq \hat{w}_{\CLASS,j'}$. Then, by statement \ref{kClassBlocks} of Lemma \ref{relateClassifierCanonical}, nodes $\hat{v}$ and $\hat{w}$ transmit in different transmission blocks of phase $P_{j'}$, and it follows that the local round at $v$ in which node $\hat{v}$ transmits is different than the local round at $w$ in which $\hat{w}$ transmits.
		%			We verify that the conditions of the induction hypothesis hold, i.e., there exists a $k \in \{0,\ldots,r\}$ such that:
		%			\begin{itemize}
		%				\item $\hat{v}$ and $\hat{w}$ are both awake in round $k$
		%				\item \bf $\hat{i} = k - \max\{t_{\hat{v}},t_{\hat{w}}\}$
		%				\item \bf $\DHist{\hat{v}}{\hat{i}} \neq \DHist{\hat{w}}{\hat{i}}$
		%			\end{itemize}
		%			By the induction hypothesis, $\hat{v} \not\equiv_{\hat{i}} \hat{w}$ in the execution of the canonical DRIP. Then $\hat{v}$ and $\hat{w}$ transmit in different transmission blocks in their local execution of the DRIP, which means that the corresponding transmission blocks at $v$ and $w$ are different.
	\end{itemize}
	In both cases above, we proved that $w$'s local round in phase $P_{j'}$ in which $\hat{w}$ transmits is different than $v$'s local round in phase $P_{j'}$ in which $\hat{v}$ transmits. This concludes the proof of Claim \ref{diffBehaviour}.

	Finally, to complete the induction step, we consider two cases that cover all possible scenarios in which $\mathcal{H}_{v,D}[i] \neq \mathcal{H}_{w,D}[i]$.
	
	%			{\bf Suppose that $w$ has no neighbours in $C$.} Consider phase $P_{i-1}$ in $v$'s execution of the canonical DRIP. By Proposition \ref{prop:eachTransmits}, there exists an $h \in \{1,\ldots,2\sigma+1\}$ and $k \in \{1,\ldots,m_{i-1}\}$ such that $\hat{v}$ transmits in a round corresponding to the $h$'th round of the $k$'th transmission block of phase $P_{i-1}$ in $v$'s execution of the canonical DRIP. This is local round number $r_{i-1} + (k-1)(2\sigma+1) + h$ at $v$. By Lemma \ref{lem:historyPhasej}, since $w$ has no neighbours, no node transmits in a round corresponding to the $h$'th round of the $k$'th transmission block of phase $P_{i-1}$ in $w$'s execution of the canonical DRIP. In particular, this means $\mathcal{H}_{v,\canonDRIP}[r_{i}+(k-1)(2\sigma+1)+h] \neq \mathcal{H}_{w,\canonDRIP}[r_{i}+(k-1)(2\sigma+1)+h]$, which implies that $\canonHist{v}{r_i} \neq \canonHist{w}{r_i}$, i.e., $v \not\equiv_{i} w$, as desired. 
	
	\begin{itemize}
		\item {\bf Suppose that, for some $z \geq 1$, node $v$ has neighbours $v_1,\ldots,v_z$ that transmit $M_1,\ldots,M_z$, respectively, during $v$'s local round $i$ in the execution of $D$. Suppose that, in the execution of $D$, node $w$ has no neighbour that transmits during $w$'s local round $i$, or, has exactly one neighbour that transmits during $w$'s local round $i$, and transmits a message $M'$ such that $M' \neq M_x$ for some $x \in \{1,\ldots,z\}$.}

%		\item {\bf $\mathcal{H}_{v,D}[i] = M$ for some message $M$, and either $\mathcal{H}_{w,D}[i] = (\emptyset)$ or $\mathcal{H}_{w,D}[i] = (M')$ for some message $M' \neq M$.} 
		
		Let $\hat{v}$ be a neighbour of $v$ that transmits a message during $v$'s local round $i$ in the execution of $D$ such that the transmitted message $M$ is not equal to a message transmitted by any neighbour of $w$. Notice that Claim \ref{diffBehaviour} applies to each neighbour $\hat{w}$ of $w$ in $G$. In particular, this means that, for each neighbour $\hat{w}$ of $w$, there exists a $j' \geq 1$ such that, in the execution of $\canonDRIP$, the local round in phase $P_{j'}$ at node $w$ in which $\hat{w}$ transmits is different than the local round in phase $P_{j'}$ at node $v$ in which $\hat{v}$ transmits. We denote by $j_{max}$ the maximum such $j'$ taken over all neighbours $\hat{w}$ of $w$. 
		
		 Let $h \in \{1,\ldots,2\sigma+1\}$ and $k \in \{1,\ldots,numClasses_{G,j_{max}}\}$ such that $\hat{v}$ transmits in a round corresponding to the $h$'th round of the $k$'th transmission block of phase $P_{j_{max}}$ in $v$'s execution of the canonical DRIP. By the choice of $j_{max}$ and Lemma \ref{lem:staySeparate}, we conclude that no neighbour of $w$ transmits in the $h$'th round of the $k$'th transmission block of phase $P_{j_{max}}$ at $w$. In particular, this means $\mathcal{H}_{v,\canonDRIP}[r_{j_{max}-1}+(k-1)(2\sigma+1)+h] \neq \mathcal{H}_{w,\canonDRIP}[r_{j_{max}-1}+(k-1)(2\sigma+1)+h]$, which implies that $\canonHist{v}{r_{j_{max}}} \neq \canonHist{w}{r_{j_{max}}}$, as desired. Setting $j=j_{max}+1$ gives the desired result.

		\item {\bf Suppose that, for some $z \geq 2$, node $v$ has neighbours $v_1,\ldots,v_z$ that all transmit the same message $M$ during $v$'s local round $i$ in the execution of $D$. Suppose that, in the execution of $D$, $w$ has exactly one neighbour that transmits during $w$'s local round $i$, and transmits the message $M$ during this round.}
		
		Consider neighbours $v_1$ and $v_2$ of $v$ that transmit the message $M$ during $v$'s local round $i$ in the execution of $D$. 
		
		First, for each of $v_1$ and $v_2$, notice that Claim \ref{diffBehaviour} applies to each neighbour $\hat{w}$ of $w$ that does not transmit during $w$'s local round $i$. In particular, for each $\alpha \in \{1,2\}$, for each neighbour $\hat{w}$ of $w$ that does not transmit during $w$'s local round $i$, there exists a $j_\alpha' \geq 1$ such that, in the execution of $\canonDRIP$, the local round in phase $P_{j_\alpha'}$ at node $w$ in which $\hat{w}$ transmits is different than the local round in phase $P_{j_\alpha'}$ at node $v$ in which $v_\alpha$ transmits. For each neighbour $\hat{w}$ of $w$ that does not transmit during $w$'s local round $i$, we take the maximum of $j_1'$ and $j_2'$, and then denote by $j_{max}$ the maximum taken over all such $\hat{w}$. By Lemma \ref{lem:staySeparate}, we have shown that, for each neighbour $\hat{w}$ of $w$ that does not transmit during $w$'s local round $i$ in the execution of $D$, the local round in phase $P_{j_{max}}$ at node $w$ in which $\hat{w}$ transmits is different than the local round in phase $P_{j_{max}}$ at node $v$ in which $v_1$ or $v_2$ transmit.
		
		Next, let $\breve{w}$ be the neighbour of $w$ that transmits during $w$'s local round $i$ in the execution of $D$. Observe that, since we assume that $v_1$, $v_2$, and $\breve{w}$ transmit $M$ during this round, we cannot apply Claim \ref{diffBehaviour}. Instead, we consider the possible cases for when $v_1$, $v_2$, and $\breve{w}$ transmit during phase $P_{j_{max}}$ in the execution of the canonical DRIP. Define the following:
		\begin{itemize}
			\item Let $h_1 \in \{1,\ldots,2\sigma+1\}$ and $k_1 \in \{1,\ldots,numClasses_{G,j_{max}}\}$ such that $v_1$ transmits in a round corresponding to the $h_1$'th round of the $k_1$'th transmission block of phase $P_{j_{max}}$ in $v$'s execution of the canonical DRIP.
			\item Let $h_2 \in \{1,\ldots,2\sigma+1\}$ and $k_2 \in \{1,\ldots,numClasses_{G,j_{max}}\}$ such that $v_2$ transmits in a round corresponding to the $h_2$'th round of the $k_2$'th transmission block of phase $P_{j_{max}}$ in $v$'s execution of the canonical DRIP.
			\item Let $h_{\breve{w}} \in \{1,\ldots,2\sigma+1\}$ and $k_{\breve{w}} \in \{1,\ldots,numClasses_{G,j_{max}}\}$ such that $\breve{w}$ transmits in a round corresponding to the $h_{\breve{w}} $'th round of the $k_{\breve{w}} $'th transmission block of phase $P_{j_{max}}$ in $w$'s execution of the canonical DRIP.
		\end{itemize}
		We consider two cases:
		\begin{itemize}
			\item $h_1 = h_2$ and $k_1 = k_2$
			
			It follows that a collision occurs in the $h_1$'th round of the $k_1$'th transmission block of phase $P_{j_{max}}$ in $v$'s execution of the canonical DRIP. But we already showed that, for all neighbours $\hat{w} \neq \breve{w}$ of $w$, the local round in phase $P_{j_{max}}$ at node $w$ in which $\hat{w}$ transmits is different than the local round in phase $P_{j_{max}}$ at node $v$ in which $v_1$ transmits. In particular, the node $\breve{w}$ is the only possible neighbour of $w$ that might transmit in the $h_1$'th round of the $k_1$'th transmission block of phase $P_{j_{max}}$ in $w$'s execution of the canonical DRIP. This means that a collision will not happen in this round of $w$'s execution of the canonical DRIP, so $\mathcal{H}_{v,\canonDRIP}[r_{j_{max}-1}+(k_1-1)(2\sigma+1)+h_1] \neq \mathcal{H}_{w,\canonDRIP}[r_{j_{max}-1}+(k_1-1)(2\sigma+1)+h_1]$, which implies that $\canonHist{v}{r_{j_{max}}} \neq \canonHist{w}{r_{j_{max}}}$. Setting $j=j_{max}+1$ gives the desired result.
			\item $h_1 \neq h_2$ or $k_1 \neq k_2$
			
			In this case, we cannot have $h_{\breve{w}} = h_1 = h_2$ and $k_{\breve{w}} = k_1 = k_2$. Without loss of generality, assume that $h_{\breve{w}} \neq h_1$ or $k_{\breve{w}} \neq k_1$. In particular, this means that the node $\breve{w}$ does not transmit in the $h_1$'th round of the $k_1$'th transmission block of phase $P_{j_{max}}$ in $w$'s execution of the canonical DRIP. Moreover, we already showed that, for all other neighbours $\hat{w} \neq \breve{w}$ of $w$, the local round in phase $P_{j_{max}}$ at node $w$ in which $\hat{w}$ transmits is different than the local round in phase $P_{j_{max}}$ at node $v$ in which $v_1$ transmits. This means that  $\mathcal{H}_{w,\canonDRIP}[r_{j_{max}-1}+(k_1-1)(2\sigma+1)+h_1] = (\emptyset)$, and we know that $\mathcal{H}_{v,\canonDRIP}[r_{j_{max}-1}+(k_1-1)(2\sigma+1)+h_1] \neq (\emptyset)$ due to $v_1$'s transmission. It follows that $\canonHist{v}{r_{j_{max}}} \neq \canonHist{w}{r_{j_{max}}}$. Setting $j=j_{max}+1$ gives the desired result.
		\end{itemize}

%			
%		\item {\bf $\mathcal{H}_{v,D}[i] = (*)$ and $\mathcal{H}_{w,D}[i] = (M)$ for some message $M$.} Let $x$ and $y$ be two neighbours of $v$ that transmit during $v$'s local round $i$ in the execution of $D$.
%		
%	Let $\hat{w}$ be an arbitrary neighbour of $w$ in $C$.

	\end{itemize}
\end{proof}

Lemma \ref{lem:alsoDistinguish} shows that the canonical DRIP breaks symmetry among nodes at least as well as any other DRIP. Using this fact, we show that the canonical DRIP for any feasible configuration $G$ can be used to solve leader election in $G$. 

\begin{theorem}\label{DedicatedLE}
	For any configuration $G$, if $G$ is feasible then there is a $O(n^2\sigma)$-round dedicated distributed leader election algorithm $(\canonDRIP,f_G)$ for $G$.
\end{theorem}
\begin{proof}
	By Lemma \ref{lem:existsPatient}, if a configuration $G$ is feasible, then there exists a patient DRIP $D_{pat}$ and a decision function $f_{pat}$ such that $(D_{pat},f_{pat})$ solves leader election in $G$. Suppose that node $x$ is chosen as leader by this algorithm, i.e., $f_{pat}(\mathcal{H}_{x,D_{pat}}[0 \ldots done_{x,D_{pat}}]) = 1$ and $f_{pat}(\mathcal{H}_{v,D_{pat}}[0 \ldots done_{v,D_{pat}}]) = 0$ for all $v \neq x$. As $f_{pat}$ is a well-defined function, we must have $\mathcal{H}_{x,D_{pat}}[0 \ldots done_{x,D_{pat}}] \neq \mathcal{H}_{v,D_{pat}}[0 \ldots done_{v,D_{pat}}]$ for all $v \neq x$. By Lemma \ref{lem:alsoDistinguish}, it follows that $\mathcal{H}_{x,\canonDRIP}[0 \ldots done_{x,\canonDRIP}] \neq \mathcal{H}_{v,\canonDRIP}[0 \ldots done_{v,\canonDRIP}]$ for all $v \neq x$. So, by defining a decision function $f_G$ by $f_G(\mathcal{H}_{x,\canonDRIP}[0 \ldots done_{x,\canonDRIP}]) = 1$ and $f_G(\mathcal{H}_{v,\canonDRIP}[0 \ldots done_{v,\canonDRIP}]) = 0$ for all $v \neq x$, we see that $(\canonDRIP,f_G)$ solves leader election in $G$. By Lemma \ref{CanonicalComplexity}, we know that each node will terminate its execution of \canonDRIP in $O(n^2\sigma)$ rounds.
\end{proof}

Finally, we are now ready to complete the proof that \classifier outputs ``Yes" when given a feasible configuration $G$ as input. The idea is that, if $G$ is feasible, then \canonDRIP can be used to elect some leader. When \classifier is executed with input $G$, this leader will eventually be placed in its own equivalence class, and \classifier will output ``Yes" and terminate.

\begin{lemma}\label{NoCorrectness}
	For any configuration $G$, if $G$ is feasible then \classifier outputs ``Yes" when executed with $G$ as input.
\end{lemma}
\begin{proof}
	If $G$ is feasible, then, by Theorem \ref{DedicatedLE}, after executing the canonical DRIP \canonDRIP, there is at least one node $x$ such that $f_G(\mathcal{H}_{x,\canonDRIP}[0 \ldots done_{x,\canonDRIP}]) = 1$ and $f_G(\mathcal{H}_{v,\canonDRIP}[0 \ldots done_{v,\canonDRIP}]) = 0$ for all $v \neq x$. It follows that $\mathcal{H}_{x,\canonDRIP}[0 \ldots done_{x,\canonDRIP}] \neq \mathcal{H}_{v,\canonDRIP}[0 \ldots done_{v,\canonDRIP}]$ for all $v \neq x$. From the definition of the canonical DRIP \canonDRIP, the value of $done_{v,\canonDRIP}$ is the same for all $v$ in $G$: there exists some $j$, denoted by $j_{term}$, such that $\mathcal{L}_j[1] = ``terminate"$, and all nodes will terminate in local round $r_{j-1}+1$. So, $done_{v,\canonDRIP} = r_{j_{term}-1}+1$, and we have $\mathcal{H}_{x,\canonDRIP}[0 \ldots r_{j_{term}-1}] \neq \mathcal{H}_{v,\canonDRIP}[0 \ldots r_{j_{term}-1}]$ for all $v \neq x$. By Lemma \ref{lem:diffClassesHistory}, it follows that $x_{\CLASS,j_{term}} \neq v_{\CLASS,j_{term}}$ for all $v \neq x$, i.e., the condition on line \ref{yescondition} of \classifier is true after the execution of $\partitioner(G_{aug},j_{term}-1)$. In the next step, \classifier will output ``Yes" and terminate.
\end{proof}

Putting together Lemmas \ref{ClassifierComplexity}, \ref{YesCorrectness} and \ref{NoCorrectness} completes the analysis of \classifier.
\begin{theorem}
	There is a $O(n^3\Delta)$-round centralized algorithm that, when provided as input any configuration $G$ with maximum node degree $\Delta$, decides whether or not $G$ is feasible.
\end{theorem}